\documentclass[a4paper,USenglish,cleveref,autoref]{lipics-v2021}
\hideLIPIcs
\usepackage{tikz}
\usetikzlibrary{calc}

\usepackage[utf8]{inputenc}

 \newcommand\DGP{\textsc{Dense Graph Partition}}
  \newcommand\MDGP{\textsc{Max Dense Graph Partition}}
 
 \newcommand\UC{\textsc{Min UnCut}}

\newcommand{\defprob}[3]{
\begin{center}
{\begin{minipage}{.95\textwidth}
\noindent{\sc #1}\\\nopagebreak
{\bf Input:} #2 \\\nopagebreak
{\bf Question:} #3
\end{minipage}}
\end{center}
}

\newcommand{\commente}[1]{}

\newcommand{\dispFig}[1]{#1}

\newcommand{\convexpath}[2]{
	[   
	create hullnodes/.code={
		\global\edef\namelist{#1}
		\foreach [count=\counter] \nodename in \namelist {
			\global\edef\numberofnodes{\counter}
			\node at (\nodename) [draw=none,name=hullnode\counter] {};
		}
		\node at (hullnode\numberofnodes) [name=hullnode0,draw=none] {};
		\pgfmathtruncatemacro\lastnumber{\numberofnodes+1}
		\node at (hullnode1) [name=hullnode\lastnumber,draw=none] {};
	},
	create hullnodes
	]
	($(hullnode1)!#2!-90:(hullnode0)$)
	\foreach [
	evaluate=\currentnode as \previousnode using \currentnode-1,
	evaluate=\currentnode as \nextnode using \currentnode+1
	] \currentnode in {1,...,\numberofnodes} {
		-- ($(hullnode\currentnode)!#2!-90:(hullnode\previousnode)$)
		let \p1 = ($(hullnode\currentnode)!#2!-90:(hullnode\previousnode) - (hullnode\currentnode)$),
		\n1 = {atan2(\y1,\x1)},
		\p2 = ($(hullnode\currentnode)!#2!90:(hullnode\nextnode) - (hullnode\currentnode)$),
		\n2 = {atan2(\y2,\x2)},
		\n{delta} = {-Mod(\n1-\n2,360)}
		in 
		{arc [start angle=\n1, delta angle=\n{delta}, radius=#2]}
	}
	-- cycle
}

\title{Dense Graph Partitioning on sparse and dense graphs}
\author{Cristina Bazgan}{Universit\'e   Paris-Dauphine, PSL Research University, CNRS,UMR  7243, LAMSADE, 75016 Paris, France \and \url{https://www.lamsade.dauphine.fr/~bazgan/}}{cristina.bazgan@dauphine.fr}{https://orcid.org/0000-0002-5460-6222}{}

 \author{Katrin Casel}{Hasso Plattner Institute, University of Potsdam,  Germany \and \url{https://hpi.de/friedrich/people/katrin-casel.html}}{Katrin.Casel@hpi.de}{https://orcid.org/0000-0001-6146-8684}{}

 \author{Pierre Cazals}{Universit\'e   Paris-Dauphine, PSL Research University, CNRS,UMR  7243, LAMSADE, 75016 Paris, France}{pierre.cazals@dauphine.eu}{https://orcid.org/0000-0002-7681-476X}{}

 \Copyright{Cristina Bazgan, Katrin Casel, and Pierre Cazals}


\authorrunning{C. Bazgan, K. Casel and P. Cazals}
 

\begin{document}

 \nolinenumbers
\maketitle              

\begin{abstract}
We consider the problem of partitioning a graph into a non-fixed number of  non-overlapping subgraphs of maximum density. The density of a partition is the  sum of the densities of the subgraphs, where the density of a subgraph is  its average degree, that is, the ratio of its number of edges and its number of vertices. This problem, called Dense Graph Partition, is known to be NP-hard on general graphs and polynomial-time solvable on trees, and polynomial-time  2-approximable.

In this paper we study the restriction of Dense Graph Partition to particular sparse and dense graph classes. In particular, we prove that it is NP-hard on dense bipartite graphs as well as on cubic graphs. 
On dense graphs on $n$ vertices, it is polynomial-time solvable on graphs with minimum degree $n-3$ and NP-hard  on $(n-4)$-regular graphs. We prove that it is  polynomial-time  $4/3$-approximable  on cubic graphs and admits an efficient polynomial-time approximation scheme on graphs of minimum degree $n-t$ for any constant $t\geq 4$. 
\end{abstract}

\section{Introduction}
The research around communities in social networks can be seen as a contribution to the well establish
research of clustering and graph partitioning. Graph partitioning problems have
been intensively studied with various measures in order to evaluate clustering quality, see e.g.~\cite{New2004,Sch2007,For2010,BulucMSS016}  for an overview. In the context of social networks, a ‘community’ is a collection of individuals who are relatively well connected compared to other parts of the social network graph . A ‘community structure’ then corresponds to a partition of the whole social network into communities.

We consider a classical definition of the density of a (sub)graph  (see, for example,  \cite{bib:density:goldberg1984finding, KS09, bib:density:darlay2012DENSE}) given by its average degree, that is, the ratio between its number of edges and its number of vertices.  
%The density of a partition  is the sum of the densities of all its parts. 
For this definition of density, there are several papers on finding the densest subgraph. This problem was shown solvable in polynomial time by Goldberg~\cite{bib:density:goldberg1984finding} but if the size of the subgraph is a part on the input, the problem called \textsc{$k$-Densest Subgraph} becomes NP-hard even restricted to bipartite or chordal graphs~\cite{bib:density:corneil1984clustering}. The approximability of \textsc{$k$-Densest Subgraph} was also  studied, see~\cite{Khot06,FeigePK01,BhaskaraCCFV10}. 
% Goldberg [9] (also see [14]) proved that the densest subgraph problem can be solved optimally in polynomial time.  Research results was also done on densest subgraph problem  subject to size constraints.
% In the k-densest subgraph problem the objective is to find a densest
% subgraph induced over precisely k vertices. The problem  was shown to be NP-hard by [7], Khot [11] shows that that there does
% not exist any PTAS for the k-densest subgraph problem under a reasonable complexity
% assumption. In the direction of upper bounds, [7] provide an algorithm
% with an approximation factor of O(nθ) where θ < 1/3. More recently, [3] give an
% algorithm for the k-densest subgraph with an approximation factor of O(n1/4).

In this paper, we study the problem \MDGP{} that models finding a community structure, that is, finding a dense \emph{partition}. More precisely, given an undirected graph $G$, we aim to find a partition $\mathcal{P} = \{ V_1, \dots, V_k\}$, $k \geq 1$, of the vertices of $G$, such that sum of the densities of the subgraphs $G[V_i]$ is maximized. We denote the sum of the densities of the subgraphs $G[V_i]$ by $d(\mathcal{P})$, and call this the density of the partition $\mathcal{P}$.

Note that the general concept of a community structure does not put any restriction on the
number of communities. We therefore address the problem \MDGP{} of finding a partition of maximum density, without fixing the number of classes of
the partition. 
%In this paper we address the problem \MDGP{} of finding a partition of maximum density, without fixing the number of classes of the partition. 
Indeed, when the number of classes is given, the problem is a generalization of a partition into $k$ cliques.  
By not fixing the number of classes, \MDGP{} differs from partitioning into cliques: observe that while there exists a partition into exactly $k$ sets of density $(n-k)/2$ if and only if the input graph can be partitioned into $k$ cliques (see Lemma~\ref{lemmacomplete}), there can be a partition into less than $k$ sets with a density even higher than $(n-k)/2$ even if the input cannot be partitioned into $k$ cliques. As an example, consider a complete graph of an even number $n$ of vertices and turn four of the vertices into an independent set by removing all edges among them. The resulting graph cannot be partitioned into 3 cliques (at least one set contains two of the four independent vertices), but it has a partition into two sets of equal cardinality  with density $(n-2)/2 -4/n$.

Darley et al.~\cite{bib:density:darlay2012DENSE} studied   \MDGP, and its complement  {\sc Min Sparse Graph Partition}. They defined  the  sparsity of a partition $\mathcal{P}$ as $F(\mathcal{P})= \frac{|\mathcal{P}|}{2} + d(\mathcal{P})$ and the  problem {\sc Min Sparse Graph Partition} as finding a partition of a given undirected graph $G$ such that the sparsity of the partition is minimized. Observe that  \MDGP{} and  {\sc Min Sparse Graph Partition} are dual in the sense that solving the first one on a graph $G$ is the same as solving the second one on the complement of $G$. In~\cite{bib:density:darlay2012DENSE} it is shown that both problems are NP-complete, and that there is no constant factor approximation for  {\sc Min Sparse Graph Partition} unless $P=NP$.  Moreover, a polynomial time algorithm for \MDGP{} on trees is given.
We point out that their proof of NP-completeness is a polynomial-time reduction from \textsc{$k$-Coloring}. By construction, the same reduction when starting from \textsc{3-Coloring} on graphs of degree at most 4 (proved NP-complete in~\cite{GJS1976}) yields as instance of \MDGP{} a graph on $n$ vertices and of minimum degree greater than~$n-4n^{4/5}$. Thus it follows that  \MDGP{} is NP-complete restricted to graphs of minimum degree~$n-4n^{4/5}$.

Aziz et al.~\cite{bib:density:aziz2015welfare} studied the problem \textsc{Fractional Hedonic Game}, and more particularly the \textsc{Max Utilitarian Welfare} problem as the simple symmetric version of the game defined as follows. Let $N$ be a set of agents, the utility of $i \in N$ in a coalition $S \subseteq N$ is $u_i(S) = \tfrac 1{|S|}{\sum_{j \in S}u_i(j)}$ where $u_i(j)$ is such that $u_i(j) \in \{0,1\}$ for a simple game and $u_i(j) = u_j(i)$ for a symmetric one. For \textsc{Max Utilitarian Welfare} one tries to find a partition~$C$ of~$N$ into coalitions that maximizes $\sum_{S \in C}\sum_{i \in S}u_i(S)$. This game can be seen as a graph~$G$ where agents are  vertices and there is an edge between two agents $i$ and $j$ if and only if $u_i(j)=1$.
In this context, $u_i(S) = \tfrac 1{|S|}{\sum_{j \in S}u_i(j)} = \tfrac 1{|S|}deg_{G[S]}(i)$. We deduce that $\sum_{S \in C}\sum_{i \in S}u_i(S) =\tfrac 1{|S|} \sum_{S \in C}\sum_{i \in S} deg_{G[S]}(i) = \tfrac 1{|S|}\sum_{S \in C}  {2|E(S)|} = 2 \cdot d(C)$. Hence, the problems \textsc{Max Utilitarian Welfare} and \MDGP{} are equivalent to within a constant, which means that the 2-approximation for the former given in~\cite{bib:density:aziz2015welfare} directly translates to the latter.
%In~\cite{bib:density:aziz2015welfare} it is shown that \textsc{Max Utilitarian Welfare} admits a polynomial-time 2-approximation based on \textsc{Maximum Matching}. This result directly transfers to \MDGP.

\medskip
\noindent
\emph{Our contributions.}  The following overview summarises the results achieved in this paper concerning \MDGP{} (MDGP). 
\begin{itemize}
\item MDGP is trivially solvable on graphs of maximum degree 2, we prove its NP-hardness for 3-regular (cubic) graphs. 

\item We establish that on bipartite complete graphs an optimal partition  consists of one part, that is the whole graph. Moreover if the size of the two independent sets are relatively prime numbers then this optimal solution is unique. We use this result to show that MDGP is NP-hard on dense bipartite graphs.

\item MDGP is trivial on complete graphs since the optimal solution is the whole graph as one part of the partition. Moreover, as we previously explained, it is NP-hard on graphs of minimum degree $n-4n^{4/5}$.  We show that for graphs of minimum degree $\geq n-3$, the problem is solvable in polynomial time and any optimal solution has two parts.  Moreover on $(n-4)$-regular graphs, the problem becomes NP-hard. 

\item We further give improves on the 2-approximation for MDGP on general graphs~\cite{bib:density:aziz2015welfare} for specific sparse and dense graph classes. In particular, we show that  MDGP admits a polynomial-time $4/3$-approximation  on cubic graphs. Moreover we establish a polynomial-time $\frac{n-1}{\delta+1}$-approximation, where $\delta$ is the minimum degree of the input graph (note that this improves on the ratio of 2 for all $\delta >\frac{n-3}{2}$). Also, we give an eptas (i.e.~a $(1+\varepsilon)$-approximation for any $\varepsilon >0$) on graphs of minimum degree $n-t$ for any constant $t\geq 4$.
\end{itemize}

Our paper is organized as follows. Notations and  formal definitions are given in
Section~\ref{sec2}. The study of (dense) bipartite graphs is established in Section~\ref{sec3}. Section~\ref{sec4} presents the results on cubic graphs. In Section~\ref{sec5} we study
dense graphs. %, that is graphs of minimum degree at least $n-4$. 
 Some conclusions are given at the end of the paper.

\section{Preliminaries}\label{sec2}
 In this paper we assume that all graphs are undirected, without loops or  multiple edges, and not necessary connected. We use $G=(V,E)$ to denote an undirected  graph with a set~$V$ of vertices and a set $E$ of edges.  We use $|V|$ to denote the number of vertices in $G$, i.e., the order of $G$, and we use $|E|$ to denote the number of edges in $G$, i.e., the size of $G$. We denote by $deg_G(v)$ the degree of $v\in V$ in $G$ that is the number of edges incident to $v$ and by  $D_{G}(i)$   the set of vertices of degree $i$ in $G$.   
The maximum degree of $G$, denoted by $\Delta(G)$, is the degree of the vertex with the greatest number of edges incident to it. The minimum degree of $G$, denoted by $\delta(G)$, is the degree of the vertex with the least number of edges incident to it. For any vertex $v\in V$, $N_G(v)$ is the set of neighbors of~$v$ in $G$ and $N_G[v]= N_G(v)\cup\{v\}$.  Moreover, $N_G(S)=\bigcup_{v \in S} N_G(v)$. For a  graph $G=(V,E)$  and a subset $S\subseteq V$ we denote by $E(S)$ the set of the edges of $G$ with both endpoints in~$S$. For a given  partition $\{A,B\}$  of $V$, we  denote  by $E(A,B) = \{uv \in E:~u \in A,~ v \in B\}$. Further, $G[S]$ denotes the graph induced by $S$, defined as $G[S]=(S, E(S))$.
   
 A triangle graph is  the cycle graph $C_{3}$ or  the complete graph  $K_{3}$.
 A diamond graph has 4 vertices and 5 edges, it consists of a complete graph $K_{4}$ minus one edge. A graph is called cubic  if all its vertices are of degree three. A  graph is bipartite if its vertices can be partitioned into two sets $A$ and $B$ such that 
 every edge connects a vertex in $A$ to one in $B$. A complete bipartite graph  is a special kind of bipartite graph where every vertex of $A$ is connected to every vertex in $B$. A graph  on $n$ vertices is $\delta$-dense if its minimum degree is at least $\delta n$. A set of instances is called  dense if there is a constant $\delta >0$ such that all instances in this set are $\delta$-dense (this notion was introduced in \cite{AroraKK95} and called everywhere-dense).  

The density $d(G)$ of a graph $G=(V,E)$ is the ratio between the number of edges and the number of vertices in $G$, that is, $d(G)=\tfrac{|E|}{|V|}$. Moreover, for $S\subseteq V$, $d(S)=d(G[S])=\tfrac{|E(S)|}{|S|}$. We use $\mathcal{P}$ to denote a partition of the set $V$
 of vertices of $G$, that is,  $\mathcal{P} = \{ V_1, \dots, V_k\}$, where $\cup_{i=1}^k V_i = V$, and $V_i \cap V_j = \emptyset$ for each $i, j \in \{1, \dots, k\}$. Then the density of a partition $\mathcal{P}$ of $G$ is defined as  $d(\mathcal{P})=\sum_{i=1}^k d(G[V_i])$, where $G[V_i]$ is the subgraph of $G$ induced by the subset $V_i$ of vertices, that is, $G[V_i]= (V_i, E_i)$, $E_i = \{\{u,v\} : \{u,v\} \in E \land u,v\in V_i\}$.

We study the problem of finding a partition $\mathcal{P} = \{ V_1, \dots, V_k\}$ of a given graph~$G$, such that $k \geq 1$  
and that, among all such partitions,  $d(\mathcal{P})$ is maximized. We refer to this problem as {\sc Max Dense Graph Partition} and we define its decision version as follows.

\begin{center}
\fbox{\begin{minipage}{.95\textwidth}
\noindent{\sc Dense Graph Partition}\\\nopagebreak
{\bf Input:} An undirected   graph $G=(V,E)$, a positive rational number $r$. \\\nopagebreak
{\bf Question:}  Is there a partition $\mathcal{P}$ such that $d(\mathcal{P}) \geq r$ ?
\end{minipage}}
\end{center}

Given an optimization problem in NPO and an instance $I$ of this problem, we denote by $|I|$ the size of $I$, by $opt(I)$
the optimum value of $I$, and by $val(I, S)$ the value of a feasible solution $S$ of instance $I$. The performance ratio of $S$ (or
approximation factor) is $r(I, S) = \max  \{\frac{val(I,S)}{opt(I)},\frac{opt(I)}{val(I,S)}\}\geq 1$.  
For a function $f$, an algorithm is an $f(|I|)$-approximation, if for every instance $I$ of the problem, it returns a solution $S$
such that $r(I, S) \leq  f (|I|)$. Moreover if the algorithm runs in polynomial time in  $|I|$, then this algorithm gives a polynomial-time $f(|I|)$-approximation.
We consider in this paper only polynomial time algorithms. When $f$ is a constant $\alpha$, the problem is polynomial-time $\alpha$-approximable. When $f=1+\varepsilon$, for any $\varepsilon >0$, the problem admits a polynomial-time approximation scheme. When the running time of an approximation scheme is of the form $O(g(1/\varepsilon)poly(|I|)$ the problem has an efficient polynomial-time approximation scheme (eptas). 

Before we start studying specific graph classes, we observe the following helpful structural properties that hold for {\sc Dense Graph Partition} on general graphs.
\begin{remark}\label{remarkconnexe}
 We can assume that for any optimal partition  $\mathcal{P}$ and for any part $P_i \in \mathcal{P}$, $G[P_i]$ is connected, since otherwise turning each connected  component into its own part does not decrease the density.
\end{remark}

When discussing the density of a (sub)graph, it is often useful to think about how close this subgraph is to being a clique. We therefore call a pair of non-adjacent vertices in a (sub)graph a \emph{missing edge}, and use the number of such missing edges to estimate the density of the (sub)graph. With such estimations, it is easy to show that the following intuition about favouring complete graphs as communities.

 \begin{lemma}\label{lemmacomplete}
   Among all partitions of $G$ into $t\geq 2$ parts, those where the parts correspond to complete graphs, if there exists such,  have the largest density. 
 \end{lemma}
 \begin{proof}
Consider a partition of $G$ into $t$ parts $\{V_1,\ldots,V_t\}$ of size $n_1,\ldots,n_t$.  If $G[V_i]$ has $o_i$ missing edges for any $1\leq i\leq t$, then the density of this partition is $\frac{n-t}{2}- \frac{o_1}{n_1}-\ldots -\frac{o_t}{n_t}$. 

Consider a partition of $G$ into $t$ parts   of size $n_1',\ldots,n_t'$ such that each part induces a  complete graph for any $1\leq i\leq t$. Then the density of this partition is $\frac{n-t}{2}$ and thus it is larger than the density of any partition in  $t$ parts where at least one edge is missing inside $G[V_i]$ for some $1\leq i\leq t$.
\end{proof}
A direct consequence of this is the following.
\begin{lemma}
\label{lem:densMax}
Let  $G = (V,E)$ be a graph and $\mathcal{P}$ be any partition of $V$. Then $d(\mathcal{P}) \leq \frac{|V|}{2} - \frac{|\mathcal{P}|}{2}$.
\end{lemma}

\section{Dense Bipartite Graphs}\label{sec3}
\label{sec:bipartiDense}
In this section we show that \MDGP{} has a trivial solution on complete bipartite graphs.  Moreover, using this result we  show that the problem is NP-hard on dense bipartite graphs. 

In the first part, we consider a  complete bipartite graph $G_{n,m}$ with the two subsets that are independent sets of size $n$ and $m$ and we first prove the following result.
 
\begin{lemma} \label{complete_bipartite_lem}
The density $d(G_{n,m})$ of a complete bipartite  graph $G_{n,m}$ is greater than or equal to the density $d(\mathcal{P})$ of any partition $\mathcal{P}$ of $G_{n,m}$.
\end{lemma}
\begin{proof}

The density of the complete bipartite  graph $G_{n,m}=(A,B,E)$, with $|A|=n, |B|=m$ is given by $d(G_{n,m}) = \frac{nm}{n+m}$
 
It suffices to show that $d(G_{n,m})$ is greater than or equal to the density of any partition $\mathcal{P} = \{ V_1, V_2\}$ that splits the set of vertices into exactly $2$ nonempty subsets. Indeed,  if this holds and  we have a partition $\mathcal{P} = \{ V_1, \dots, V_k\}$ where $k \geq 3$, we can show recursively that $d(G_{n,m}) \geq d(G[V_1])+ d(G[V_2 \cup \dots \cup V_k]) \geq \dots \geq d(G[V_1])+ \dots + d(G[V_k])$.

We first consider a partition $\mathcal{P}_1 = \{ V_1,  V_2\}$ where  $A\subseteq V_1$. Without loss of generality we may assume that $V_2=B\setminus V_1$ contains $m_2$  vertices from $B$. Then 
\begin{equation*}
d(\mathcal{P}_1) = \frac{n(m-m_2)}{n+m-m_2} + 0 \leq \frac{nm}{n+m}
\label{density_of_P_1}
\end{equation*}
Now, consider a partition $\mathcal{P}_1 = \{ V_1,  V_2\}$ such that each of the graphs $G[V_i]$ contains at least one edge, so let $G[V_1]=G_{n_1,m_1}$ with $0<n_1<n$ and $0<m_1<m$. Then $G[V_2]=G_{n-n_1,m-m_1}$ and
\begin{equation*}
d(\mathcal{P}_1) = \frac{n_1m_1}{n_1+m_1} + \frac{(n-n_1)(m-m_1)}{n+m-n_1-m_1} = \frac{nm(n_1+m_1)-mn_1^2-nm_1^2}{(n+m-n_1-m_1)(n_1+m_1)}\,,
\end{equation*}
which yields
\begin{equation*}
d(G_{n,m})-d(\mathcal{P}_1) =\frac{(nm_1-mn_1)^2}{(n+m-n_1-m_1)(n_1+m_1)(n+m)}\geq 0
\end{equation*}
\end{proof}

It follows that an optimal solution of  any complete bipartite graph is the whole graph. From the calculations in the previous proof, we can inductively deduce the following result.
\begin{corollary}\label{bipar_cond}
For any complete bipartite graph $G=(A,B,E)$ with $|A|=n$ and $|B|=m$, a partition $\mathcal{P}=\{V_1,\dots,V_k\}$ of $A\cup B$ satisfies $d(\mathcal{P}) = \frac{nm}{n+m}$    if and only if $G[V_i]=G_{n_i,m_i}$ with $n_i\not=0$ and $m_i\not=0$ and $\frac{n_i}{m_i}=\frac{n}{m}$ for all $i\in\{1,\dots,k\}$.
\end{corollary}
 
 Consequently, for  any complete bipartite graph $G_{n,m}$, if $n$ and $m$ are relatively prime  the only 
  optimal solution of $G_{n,m}$ is the whole graph. Otherwise, several optimal solutions exist and are characterized exactly by Corollary~\ref{bipar_cond}.

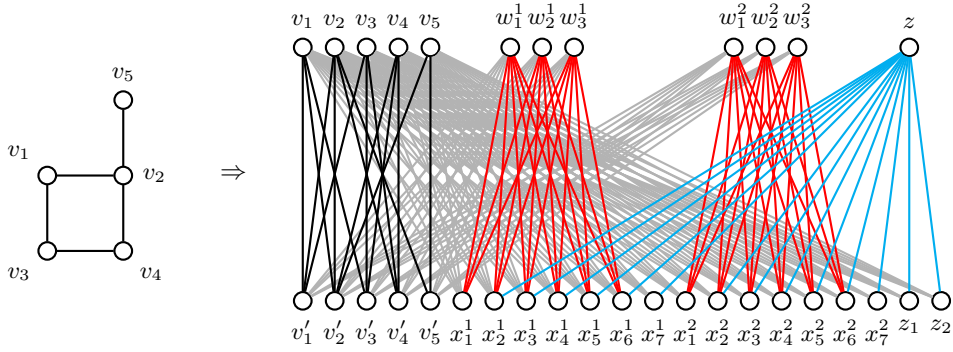
\begin{figure}
\begin{center}
\begin{minipage}[c]{0.2\linewidth}
\begin{tikzpicture}[thick,scale=1]
   \tikzset{node/.style 2 args={draw, circle,draw=black,scale=0.7, label=#1:{\large #2}}};
   \node[node={135}{\small $v_1$} ] (0) at (0,0) {}; 
\node[node={0}{\small $v_2$} ] (1) at (1,0) {};
\node[node={225}{\small $v_3$} ] (2) at (0,-1) {};
\node[node={315}{\small $v_4$} ] (3) at (1,-1) {};
\node[node={90}{\small $v_5$} ] (4) at (1,1) {};

\draw (0) -- (1);
\draw (0) -- (2);
\draw (2) -- (3);
\draw (3) -- (1);
\draw (1) -- (4);
   
\end{tikzpicture}
\end{minipage}
\begin{minipage}[c]{0.05\linewidth}
$\Rightarrow$
\end{minipage}
\begin{minipage}[c]{0.7\linewidth}
\begin{tikzpicture}[thick,scale=0.42]
   \tikzset{node/.style 2 args={draw, circle,draw=black,scale=0.7, label=#1:{\large #2}}};

\node[node={90}{\small $v_1$} ] (0) at (0,0) {}; 
\node[node={90}{\small $v_2$} ] (1) at (1,0) {};
\node[node={90}{\small $v_3$} ] (2) at (2,0) {};
\node[node={90}{\small $v_4$} ] (3) at (3,0) {};
\node[node={90}{\small $v_5$} ] (4) at (4,0) {};
\node[node={90}{\small $w_1^1$} ] (5) at (6.5,0) {};
\node[node={90}{\small $w^1_2$} ] (6) at (7.5,0) {};
\node[node={90}{\small $w^1_3$} ] (7) at (8.5,0) {};
\node[node={90}{\small $w^2_1$} ] (8) at (13.5,0) {};
\node[node={90}{\small $w^2_2$} ] (9) at (14.5,0) {};
\node[node={90}{\small $w^2_3$} ] (10) at (15.5,0) {};
\node[node={90}{\small $z$} ] (11) at (19,0) {};

\node[node={270}{\small $v'_1$} ] (12) at (0,-8) {};
\node[node={270}{\small $v'_2$} ] (13) at (1,-8) {};
\node[node={270}{\small $v'_3$} ] (14) at (2,-8) {};
\node[node={270}{\small $v'_4$} ] (15) at (3,-8) {};
\node[node={270}{\small $v'_5$} ] (16) at (4,-8) {};

\foreach \i in {1,...,7}{
        \node[node={270}{\small $x^{1}_{\i}$} ] (a\i) at (4 + \i ,-8) {};}

\foreach \i in {1,...,7}{
        \node[node={270}{\small $x^{2}_{\i}$} ] (b\i) at (4 + \i + 7,-8) {};}

\node[node={270}{\small $z_1$} ] (31) at (19,-8) {};
\node[node={270}{\small $z_2$} ] (32) at (20,-8) {};
%-
%Ec--------------------------

\foreach \i in {0,...,4}{
    \foreach \j in {1,...,7}{
        \draw (\i) edge[gray!60] (a\j);
    }
}

\foreach \i in {0,...,4}{
    \foreach \j in {1,...,7}{
        \draw (\i) edge[gray!60] (b\j);
    }
}

\foreach \j in {0,...,4}{
    \draw (\j) edge[gray!60] (31);
    \draw (\j) edge[gray!60] (32);
}

\foreach \i in {12,...,16}{
    \foreach \j in {5,...,10}{
        \draw (\i) edge[gray!60] (\j);
    }
}

%Ewx-----------------------------

\foreach \i in {5,6,7}{
    \foreach \j in {1,...,6}{
        \draw (\i) edge[red] (a\j);
    }
}

\foreach \i in {8,9,10}{
    \foreach \j in {1,...,6}{
        \draw (\i) edge[red] (b\j);
    }
}

%Ed----------------------

\draw (0) -- (12);
\draw (0) -- (13);
\draw (0) -- (14);

\draw (1) -- (12);
\draw (1) -- (13);
\draw (1) -- (15);
\draw (1) -- (16);

\draw (2) -- (12);
\draw (2) -- (14);
\draw (2) -- (15);

\draw (3) -- (13);
\draw (3) -- (14);
\draw (3) -- (15);

\draw (4) -- (13);
\draw (4) -- (16);

%Ez-----------------

\foreach \j in {2,...,7}{
    \draw (11) edge[cyan] (a\j);
}
\foreach \j in {2,...,7}{
    \draw (11) edge[cyan] (b\j);
}
\draw (31) edge[cyan] (11);
\draw (32) edge[cyan] (11);

\end{tikzpicture}
\end{minipage}
\caption{A graph $G$, instance of \textsc{Dominating Set} and the  bipartite graph $G'$ obtained from $G$, for $k=2$ and $n=5$.}
\end{center}
\end{figure}\label{bipartitefigure}
\begin{theorem} \label{bipartite}
\DGP~is NP-hard  on dense bipartite graphs. 
\end{theorem}

\begin{proof}
We give a reduction from {\sc Dominating Set}. Let $G=(V,E)$ with $V=\{v_1,\dots,v_n\}$ and an integer $k \geq 1$ be an instance of {\sc Dominating Set}. Assume without loss of generality that $G$ is connected. We first construct a bipartite graph $G'=(V_1,V_2,E')$, that is not dense, and show how solving \DGP~on it solves {\sc Dominating Set} on $G$. In a second step, we show how to make $G'$ dense maintaining the reduction. 

We construct $G'=(V_1,V_2,E')$  as follows:
\begin{itemize}
\item $V_1= V\cup \{w_i^j\colon 1\leq i\leq n-k, \  1\leq j\leq k\}\cup\{z\}$
\item $V_2=V'\cup\{x_r^j\colon 1\leq r\leq N, \ 1\leq j\leq k\}\cup\{z_i\colon 1\leq i\leq N-n\}$ where $V'=\{v_1',\dots,v_n'\}$ and $N\in \mathbb{N}$ is chosen as follows.
Let $c\in \mathbb N$ be the smallest integer such that  $c(n-k+1)-1 > n$ (note that $1\leq c\leq n$) and define $N=c(n-k+1)-1$. For this choice of $N$ it follows that the greatest common divisor of $N$ and $n-k+1$ is 1, and $n<N\leq 2n$.
\item $E'=E_d\cup E_{wx}\cup E_c  \cup E_z$ with

$E_d=\{\{v_i,v_j'\}\colon \{v_i,v_j\}\in E\}\cup \{\{v_i,v_i'\}\colon 1\leq i\leq n \}$,\\
$E_{wx}=\{\{w_i^j,x_r^j\}\colon 1\leq i\leq n-k, \ 1\leq r\leq N-1, 1\leq j\leq k \}$, \\
$E_c=\{\{w_i^j,v_s'\}\colon 1\leq i\leq n-k, 1\leq j\leq k,\ 1\leq s\leq n\}\cup\{\{v_s,x_r^j\}\colon 1\leq s\leq n, \ 1\leq r\leq N, \ 1\leq j\leq k\}$ and \\ $E_z=\{\{z,z_j\}\colon 1\leq j\leq N-n\} \cup \{\{z,x^j_r\} : 2\leq r\leq N, 1\leq j\leq k \} \cup\{\{v_i,z_j\}\colon 1\leq i\leq n, \  1\leq j\leq N-n \}$
\end{itemize}
Notice that  $G'$ is a bipartite graph with $|V_1|=n+1+ k(n-k)$ and  $|V_2|= (k+1)N$.

We show that there exits a dominating set of cardinality at most $k$ in $G$ if and only if there exists a partition  $\mathcal{P}$ of $G'$ with  $d(\mathcal{P})=(k+1)d(G_{n-k+1,N})$.

Suppose there exists a dominating set $D$ in $G$ with $|D|=k$.
Let $D=\{v_{i_1},\dots,v_{i_k}\}$  and $N'(v_{i_j})=N_G[v_{i_j}]\setminus (D\cup N_G(\{v_{i_1},\dots, v_{i_{j-1}}\})$. Define the partition $\mathcal{P}=\{P_1,\dots,P_{k+1}\}$ by:\\ $P_j=\{v_{i_j}\}\cup \{v_r'\colon v_r \in N'(v_{i_j})\} \cup \{w_r^j\colon 1\leq r \leq n-k\} \cup\{x_r^j\colon 1\leq r\leq N-|N'(v_{i_j})| \}$ for $1\leq j\leq k $ and $P_{k+1}=V_1\cup V_2\setminus (\cup_{j=1}^k P_j)$.
With this definition, $\mathcal{P}$ is clearly a partition of $V_1\cup V_2$, and each part $P_j$ contains $n-k+1$ vertices from $V_1$ and $N$ vertices from $V_2$ for each $1\leq j\leq k+1$. Further, each $P_j$ induces a complete bipartite graph $G_{n-k+1,N}$: All vertices $w_r^j$ and $x_r^j$ are connected to each other, and to all vertices in $V_2$ and $V_1$, respectively, by construction. 
Further, $v_{i_j}$ is connected in $G'$ to all vertices in $N'(v_{i_j})$; note here that in  $G'$ we connected $v_i$ to its ``copy'' $v'_i$ for all $1\leq i\leq n$, which models the case that $v_{i_j}$ dominates itself. For $P_{k+1}$, note that $z$ is adjacent to all $x_i^j$-vertices, and each $z_i$ is adjacent to all vertices in $V$. Since $D$ is a dominating set, each vertex from $V'$ is contained in some $N'(v_{i_j})$, thus  $V_2\setminus (\cup_{j=1}^k P_j)$ only contains  $x_i^j$-vertices. Also, the $P_j$ contain all $w_i^j$ vertices and hence $V_1\setminus (\cup_{j=1}^k P_j)$ only contains vertices from $V$.

Conversely, let $\mathcal{P}$ be a partition of $G'$ of density $(k+1)d(G_{n-k+1,N})$. 
Thus, Corollary~\ref{bipar_cond} implies that the vertices for each set $P\in \mathcal{P}$ induce a complete bipartite graph $G_{r,s}$ such that $\frac rs=\frac{|V_1|}{|V_2|}=\frac{k(n-k)+n+1}{(k+1)N}=\frac{n-k+1}{N}$. Since the greatest common divisor of $n-k+1$ and $N$ is one, this yields $r\geq n-k+1$ and $s\geq N$ and especially $\mathcal{P}$ can contain at most $k+1$ sets. 

For all $w_i^j$ and $w_{\ell}^{t}$, if $j \neq t$, $w_i^j$ and $w_{\ell}^{t}$ have $n$ common neighbors, and since $n < N$  there is no part $P\in \mathcal{P}$ such that $w_i^j, w_{\ell}^{t} \in P$. Moreover, for all $i,j$, $w_i^j$ and $z$ have $N-1$ common neighbors so they also cannot be in the same $P\in \mathcal{P}$. Hence, there are exactly $k+1$ parts in $\mathcal{P}$ that are complete bipartite graphs $G_{n-k+1,N}$.

For all $1 \leq j \leq k$,  denote by $P_j$ the set containing the vertices $w^j_i$ for all $1 \leq i \leq n-k$ and $P_z$ the set containing $z$. To reach  cardinality exactly $n-k+1$, $P_j\cap V_1$ has to contain exactly one vertex from $V$ for each $1\leq j\leq k$.
Further, since for any $i$, $v_i'$ is not adjacent to $z$, $V' \subseteq \cup_{j=1}^k P_j$. 
As each $P \in \mathcal{P}$ induces a complete bipartite graph in $G'$, $D= V \cap  \cup_{j=1}^k P_j$ is a set of size $k$, such that each vertex in $V'$ is adjacent to at least one vertex in $D$, so we deduce that $D$ is a dominating set of size $k$ in $G$. 

\medskip
We extend the construction of the proof to create from $G'$ a dense bipartite graph $G''=(V'',E'')$ by adding four sets of vertices $V_1^u, V_1^d, V_2^u, V_2^d$ with $|V_1^u|=|V_1^d|=kn|V_1|=kn(k(n-k)+n+1)$  and $|V_2^u|=|V_2^d|=kn|V_2|=knN(k+1)$. Further, we add edges to turn the pairs $(V_1^u, V_2^u)$, $(V_1^d, V_2^d)$, $(V_1^u, V_1)$, and $(V_2^d,V_2)$ each into complete bipartite graphs. Observe that with this construction $G''$ has $|V''|=(2kn+1)(k(n-k)+n+1)+(2kn+1)N(k+1)< 10k^2n^2$  vertices and that all vertices have degree at least $kn|V_1|\geq \frac 12 k^2 n^2\in \Theta(|V''|)$. (Note that if $k \geq \frac{n}{2}$, $G$ is a trivial yes-instance for \textsc{Dominating Set}.)

We claim that there exists a partition $\mathcal{P'}$ of $G''$ with  $d(\mathcal{P'})=(k+1)d(G_{n-k+1,N})+2kn(k+1)d(G_{n-k+1,N})$ if and only if there exists a dominating set of size $k$ for $G$.  Corollary~\ref{bipar_cond} again implies that this density for  $G''$ can only be achieved by a partition into complete bipartite graphs $G_{r,s}$ with $\frac rs=\frac{(2kn+1)(k(n-k)+n+1)}{(2kn+1)N(k+1)}=\frac{n-k+1}{N}$. The vertices in $V_1^d$ are only adjacent to vertices in $V_2^d$, and the vertices in $V_2^u$ are only adjacent to vertices in $V_1^u$. Clustering these in a ratio $\frac rs$ results in clusters containing exactly all newly added vertices, and this can be done with just two sets in total. 
What remains is to cluster the graph $G'$ into complete bipartite graphs $G_{r,s}$ such that $\frac rs=\frac{|V_1|}{|V_2|}=\frac{k(n-k)+n+1}{(k+1)N}=\frac{n-k+1}{N}$ as before.
\end{proof}

\section{Cubic Graphs }\label{sec4}
\label{sec:cubicDense}
In this section, we study \DGP~on cubic graphs, show that it remains NP-complete on this restricted graph class, but also give a polynomial time  $\frac{4}{3}$-approximation for its optimization variant \MDGP.
 We start with some general observations on the structure of communities in cubic graphs.

\begin{definition}
For $P \subseteq V$, the utility of a vertex  $v \in P$ is defined by  $u_P(v) = \frac{d(S)}{|P|}$, and the utility of $P$ is defined by $u(P)=u_P(v)$ for any $v \in P$. For a partition $\mathcal{P}= \{ V_1, \ldots, V_k\}$, the utility of a vertex $v$ in $\mathcal{P}$ is defined by  $u_{\mathcal{P}}(v)=u_{V_i}(v)$ with $i$ such that $v\in V_i$. 
\end{definition}

Considering these definitions, we can remark that:
\begin{itemize} 
\item  For any subset $P \subseteq V$, and $v,w \in P$, $u_P(v)=u_P(w)$.
\item If $P=\{v\}$ then $u_P(v)=0$.
\item  For any partition $\mathcal{P}$ of $G$, $\sum\limits_{V_i \in \mathcal{P}}d(V_i) = \sum\limits_{v \in V}u_{\mathcal{P}}(v)$. 
\end{itemize}

\begin{lemma}
\label{lem:utilPasTriDia}\label{lem:UtilTriangleCubique}\label{lem:UtilMatchCubique}
Let $G=(V,E)$ be a cubic graph without connected components that induce a $K_4$. For any partition $\mathcal{P}$ of $G$  the following holds:
\begin{itemize}
\item $u_{\mathcal{P}}(v) \leq \frac{1}{3}$ for all vertices $v\in V$
\item if $P\in \mathcal P$ is not a triangle, diamond or Case 1 in Figure~\ref{fig:utilDiamants} then $u(P) \leq \frac{1}{4}$
%\item if $P$ does not contain a triangle, then $u_\mathcal{P}(v) \leq \frac{1}{4}$
\end{itemize}
%Let $P \in \mathcal{P}$ be the part that contains $v$. 
\end{lemma}
\begin{proof}
Let  $\mathcal{P}$ be a partition of $G$, $P\in \mathcal{P}$ and $v\in P$.
Since $G$ is cubic, $d(P) \leq \frac{3|P|}{2|P|} = \frac{3}{2}$.
Then $u_{\mathcal{P}}(v) \leq \frac{3}{2|P|}$. If $|P| \geq 6$, $u_{\mathcal{P}}(v) \leq \frac{3}{2 \cdot 6} = \frac{1}{4}$. For $|P|=5$ it follows that $u_{\mathcal{P}}(v) \leq \frac{7}{25}<\frac 13$, since a cubic graph on 5 vertices cannot have more than 7 edges. Also, since there exists no $K_4$ in $G$,   the only graph on 5 vertices with 7 edges is Case 1 in Figure~\ref{fig:utilDiamants}, and all other graphs on $5$ vertices have $6$ or less edges which yields a utility of at most $\frac{6}{25}<\frac 14$.

Case analysis on the graphs of size~4 or less yields that the largest utility is achieved for~$P$ being a triangle, which gives $u_{\mathcal{P}}(v) = \frac{1}{3}$.  Further, if $P$ is not a triangle or a diamond, case analysis on the graphs of size~4 or less shows that $u_{\mathcal{P}}(v)$ is maximized when $P$ is an induced matching and its value is $\frac{1}{4}$.
\end{proof}

\begin{figure}[b]
    \centering
     \begin{minipage}[l]{0.32\textwidth}
                \begin{tikzpicture}[thick, label distance=0cm,scale=0.88]
                     \tikzset{node/.style 2 args={draw, circle,draw=black,scale=0.7, label=#1:{\large #2}}};
            \node[node={225}{$v_1$} ] (1) at (0,0) {};
            \node[node={180}{$v_2$}] (2) at (-1,1) {};
            \node[node={45}{$v_3$} ] (3) at (0,2) {};
            \node[node={315}{$v_4$} ] (4) at (1,1) {};
            \node[node={135}{$v_5$}] (5) at (-2.5,1) {};

            \draw[] (1)--(2);
            \draw[] (2)--(3);
            \draw[] (3)--(4);
            \draw[] (4)--(1);
            \draw[] (2)--(4);
            \draw[] (1)--(5);
            \draw[] (3)--(5);
            
            \fill[color=white] (-1,-1.4) rectangle (1,-0.7);
            
            \fill[color=white] (-1,3.4) rectangle (1,2.7);
            
            \draw[gray] \convexpath{1,5,3,4}{0.3 cm};
            
            \node (n1) at (-0.5,-1.3) {Case 1};
        \end{tikzpicture}
    \end{minipage} \ \ \ 
    \begin{minipage}[c]{0.23\textwidth}
        \begin{tikzpicture}[thick, label distance=0.15cm,scale=0.58]
            \tikzset{node/.style 2 args={draw, circle,draw=black,scale=0.7, label=#1:{\large #2}}};
            \clip (-2,2.5) rectangle  (1.5,-6);
            \node[node={225}{}] (1) at (0,0) {};
            \node[node={225}{}] (2) at (-1,1) {};
            \node[node={225}{}] (3) at (0,2) {};
            \node[node={225}{}] (4) at (1,1) {};
            
            \node (n1) at (0,-1) {};
            \node (n3) at (0,3) {};
             \node (n33) at (0,3) {};           
            \draw[densely dotted] (1)--(n1);
            \draw[densely dotted] (3)--(n33);
        
            \draw[] (1)--(2);
            \draw[] (2)--(3);
            \draw[] (3)--(4);
            \draw[] (4)--(1);
            \draw[] (2)--(4);
            
        %    \fill[color=white] (-1,-1.4) rectangle (1,-0.7);
            
            \draw[gray] \convexpath{3,n3}{0.37 cm};
           \fill[color=white] (-1,3.2) rectangle (1,2.7);
            
            \draw[gray] \convexpath{1,2,4}{0.37 cm};
            \node[node={225}{}] (1a) at (0,-4) {};
            \node[node={225}{}] (2a) at (-1,-3) {};
            \node[node={225}{}] (3a) at (0,-2) {};
            \node[node={225}{}] (4a) at (1,-3) {};
            
            \node (n1a) at (0,-5) {};
            \node (n4a) at (1.8,-3) {};
            \node (n3a) at (0,-1) {};            
            
            \draw[densely dotted] (1a)--(n1a);
            \draw[densely dotted] (3a)--(n3a);
            
            \draw[] (1a)--(2a);
            \draw[] (2a)--(3a);
            \draw[] (3a)--(4a);
            \draw[] (4a)--(1a);
            \draw[] (2a)--(4a);

            \draw[gray] \convexpath{4a,n4a}{0.37 cm};
             \fill[color=white] (1.8,-2) rectangle (2.5,-4);
            
        %    \fill[color=white] (-1,3.4) rectangle (2,-4.3);
            
            \draw[gray] \convexpath{1a,2a,3a}{0.37 cm};
             \node (n1a) at (0,-5.3) {Case 2};
        \end{tikzpicture}
               
    \end{minipage}
        \begin{minipage}[c]{0.23\textwidth}
        \begin{tikzpicture}[thick, label distance=0.15cm,scale=0.58]
            \tikzset{node/.style 2 args={draw, circle,draw=black,scale=0.7, label=#1:{\large #2}}};
            \clip (-2,2.5) rectangle  (1.5,-6);
            \node[node={225}{}] (1) at (0,0) {};
            \node[node={225}{}] (2) at (-1,1) {};
            \node[node={225}{}] (3) at (0,2) {};
            \node[node={225}{}] (4) at (1,1) {};
            
            \node (n1) at (0,-1) {};
            \node (n3) at (0,3) {};
             \node (n33) at (0,3) {};           
            \draw[densely dotted] (1)--(n1);
            \draw[densely dotted] (3)--(n33);
        
            \draw[] (1)--(2);
            \draw[] (2)--(3);
            \draw[] (3)--(4);
            \draw[] (4)--(1);
            \draw[] (2)--(4);
            
        %    \fill[color=white] (-1,-1.4) rectangle (1,-0.7);
            
            \draw[gray] \convexpath{3,n3}{0.37 cm};
           \fill[color=white] (-1,3.2) rectangle (1,2.7);
            \draw[gray] \convexpath{1,n1}{0.37 cm};
           \fill[color=white] (-1,-2.2) rectangle (1,-0.5);

            \draw[gray] \convexpath{2,4}{0.37 cm};
            
            \node[node={225}{}] (1a) at (0,-4) {};
            \node[node={225}{}] (2a) at (-1,-3) {};
            \node[node={225}{}] (3a) at (0,-2) {};
            \node[node={225}{}] (4a) at (1,-3) {};
            
            \node (n1a) at (0,-5) {};
            \node (n4a) at (1.8,-3) {};
            \node (n3a) at (0,-1) {};            
            
            \draw[densely dotted] (1a)--(n1a);
            \draw[densely dotted] (3a)--(n3a);
            
            \draw[] (1a)--(2a);
            \draw[] (2a)--(3a);
            \draw[] (3a)--(4a);
            \draw[] (4a)--(1a);
            \draw[] (2a)--(4a);

            \draw[gray] \convexpath{4a,n4a}{0.37 cm};
             \fill[color=white] (1.8,-2) rectangle (2.5,-4);
            
           \draw[gray] \convexpath{1a,n1a}{0.37 cm};
           \fill[color=white] (-1,-7.2) rectangle (1,-4.5);  
            
        %    \fill[color=white] (-1,3.4) rectangle (2,-4.3);
            
            \draw[gray] \convexpath{2a,3a}{0.37 cm};
             \node (n1a) at (0,-5.3) {Case 3};
        \end{tikzpicture}
               
    \end{minipage}

    \commente{
    \begin{minipage}[c]{0.3\textwidth}
        \centering
        \begin{tikzpicture}[thick, label distance=0.15cm,scale=0.9]
            \tikzset{node/.style 2 args={draw, circle,draw=black,scale=0.7, label=#1:{\large #2}}};
            \node[node={180}{$v_1$} ] (1) at (0,0) {};
            \node[node={90}{$v_2$}] (2) at (-1,1) {};
            \node[node={0}{$v_3$} ] (3) at (0,2) {};
            \node[node={270}{$v_4$} ] (4) at (1,1) {};
            
            \node (n1) at (0,-1) {};
            \node (n3) at (0,3) {};
            
            \draw[dashed] (1)--(n1);
            \draw[dashed] (3)--(n3);
            
            \draw[] (1)--(2);
            \draw[] (2)--(3);
            \draw[] (3)--(4);
            \draw[] (4)--(1);
            \draw[] (2)--(4);
            
            \draw[gray] \convexpath{1,n1}{0.3 cm};
            \fill[color=white] (-1,-1.4) rectangle (1,-0.7);
            
            \draw[gray] \convexpath{3,n3}{0.3 cm};
          %  \fill[color=white] (-1,3.4) rectangle (1,2.7);
            
            \draw[gray] \convexpath{2,4}{0.3 cm};
            \node (n1) at (0,-1.3) {Case 3};
        \end{tikzpicture}
    \end{minipage}
    }
        \caption{Different cases of \autoref{lem:UtilDiamCubique}}\label{fig:utilDiamants}
  
\end{figure}
\begin{lemma}
\label{lem:UtilDiamCubique}
Let $G$ be a cubic graph  without connected components that induce a $K_4$, and let  $v_1,v_2,v_3,v_4$ be vertices in $G$ that induce a diamond. Then $u_{\mathcal{P}}(v_1) + u_{\mathcal{P}}(v_2) + u_{\mathcal{P}}(v_3) + u_{\mathcal{P}}(v_4) \leq \frac{5}{4}$ for any partition $\mathcal{P}$ for $G$. 
\end{lemma}
\begin{proof}
Let $\mathcal{P}$ be any partition of $G$. Let $P_1 \in \mathcal{P}$ (resp. $P_2$, $P_3$ and $P_4$) be the part that contains $v_1$ (resp. $v_2$, $v_3$ and $v_4$). We distinguish several cases.

\noindent
\textbf{Case 1:} The four   vertices $v_i$ are in the same part $P_1$. If $P_1$ is a diamond, then $d(P_1) = \frac{5}{4}$ and thus $u_{\mathcal{P}}(v_1) + u_{\mathcal{P}}(v_2) + u_{\mathcal{P}}(v_3) + u_{\mathcal{P}}(v_4) = \frac{5}{4}$. If the four vertices are in a part $P_1$ with more than 4 vertices, by Lemma~\ref{lem:utilPasTriDia} the only subgraph that gives utility more than $\frac 14$ per vertex is the graph displayed as Case 1 in Figure~\ref{fig:utilDiamants}. This graph yields a utility of $\frac{7}{25}$ which gives  $u_{\mathcal{P}}(v_1) + u_{\mathcal{P}}(v_2) + u_{\mathcal{P}}(v_3) + u_{\mathcal{P}}(v_4) = \frac{28}{25}<\frac 54$.
 
\noindent
\textbf{Case 2:} Three among the four vertices of the diamond are in the same part. Then the fourth vertex has degree at most one in its part, thus by  Lemma~\ref{lem:utilPasTriDia} its utility is at most $\frac 14$. Further, also by  Lemma~\ref{lem:utilPasTriDia}, the utility of the other three vertices is at most $\frac 13$ and we conclude that $u_{\mathcal{P}}(v_1) + u_{\mathcal{P}}(v_2) + u_{\mathcal{P}}(v_3) + u_{\mathcal{P}}(v_4) \leq 1 + \frac{1}{4} = \frac{5}{4}$.
%forming a triangle are in the same part. The other one is then a vertex of degree at most one in its part, thus by Lemma~\ref{lem:utilPasTriDia} its utility is at most $\frac 14$. Then, by \autoref{lem:utilPasTriDia}, $u_{\mathcal{P}}(v_1) + u_{\mathcal{P}}(v_2) + u_{\mathcal{P}}(v_3) + u_{\mathcal{P}}(v_4) \leq 1 + \frac{1}{4} = \frac{5}{4}$.

\noindent
\textbf{Case 3:} At most two of the four vertices are together in the same part. Then the two vertices of degree three in the diamond have degree at most one in their part, thus  by  Lemma~\ref{lem:utilPasTriDia} we deduce like in Case 2 that $u_{\mathcal{P}}(v_1) + u_{\mathcal{P}}(v_2) + u_{\mathcal{P}}(v_3) + u_{\mathcal{P}}(v_4) \leq 2\frac 14+2\frac 13< \frac 54$.
%Let $v_2$ and $v_4$  be the two vertices of degree three in the induced diamond $v_1,v_2,v_3,v_4$. At most two out of the vertices of $v_1,v_2,v_3,v_4$ being in the same set, means that $v_2$ has degree at most one in~$P_2$. Symmetrically for $v_4$ and $P_4$. If $P_1$ is not equal to $P_2$, then the degree of $v_1$ in $P_1$ is at most one; symmetrically for $v_3$. Thus all of the sets $P_1$, $P_2$, $P_3$ and $P_4$ (whether some being the same or not) contain at least one vertex of degree at most one, and are thus neither triangles nor diamonds nor the graph  in Case 1 in Figure~\ref{lem:UtilDiamCubique}. \autoref{lem:utilPasTriDia} then implies that for all $v \in \cup_{i \leq 4} P_i$, $u_{\mathcal{P}}(v) \leq \frac{1}{4}$ and we conclude that $u_{\mathcal{P}}(v_1) + u_{\mathcal{P}}(v_2) + u_{\mathcal{P}}(v_3) + u_{\mathcal{P}}(v_4) \leq 1$.
\end{proof}

\begin{lemma}
\label{lem:ApproxBorneCubique}
Let $G$ be a cubic graph on $n$ vertices without connected components that induce a $K_4$, and let $D$ be the set of diamonds in $G$ and $T$ the set of triangles in $G$  that do not belong to a diamond. For any partition $\mathcal{P}$, $d(\mathcal{P}) \leq \frac{5}{4}|D| + |T| + \frac{1}{4}(n - 3|T| - 4|D|)$.
\end{lemma}
\begin{proof}
By \autoref{lem:UtilTriangleCubique}, the only vertices with utility more than $\frac 14$ are those that are in triangles, diamonds, or the unique neighbors of diamonds (in the sense of vertex $v_5$ in Case 1 of Figure~\ref{fig:utilDiamants}), and we know that the sum of the utilities of the vertices constituting a triangle is at most $3 \cdot \frac{1}{3} = 1$.
By \autoref{lem:UtilDiamCubique}, we further know that the sum of the utilities of the vertices constituting a diamond is at most $\frac{5}{4}$. The unique neighbors of diamonds have a utility of more than $\frac 14$ if and only if they are in a part isomorphic to  Case 1 of Figure~\ref{fig:utilDiamants}, which has a density of $\frac 75 < \frac{5}{4}+\frac 14$. Thus, if $S$ is the set of unique neighbors of diamonds, then the sum of the utilities of the vertices in the diamonds in $D$ and the vertices in $S$ is at most $\frac 54 |D|+\frac 14|S|$.
All remaining vertices  have a utility of at most $\frac{1}{4}$ by \autoref{lem:UtilMatchCubique}. We deduce that $d(G) \leq \frac{5}{4}|D|+\frac 14|S| + |T| + \frac{1}{4}(n - 3|T| - 4|D|-|S|)=\frac{5}{4}|D|+ |T| + \frac{1}{4}(n - 3|T| - 4|D|)$.
\end{proof}

We show that  \DGP~is NP-complete even for cubic graphs by giving a reduction from \textsc{Exact Cover By 3-Sets} where each element appears in exactly 3 sets, denoted \textsc{Restricted Exact Cover By 3-Sets},  known to be NP-hard by~\cite{Gonzalez85}.

\defprob{Restricted Exact Cover By 3-Sets (RX3C)}{A set $X$ of elements with $|X| = 3q$ and a collection $C$ of 3-element subsets of $X$ where each element appears in exactly 3 sets.}{Does $C$ contain an exact cover for $X$, i.e. a subcollection $C' \subseteq C$ such that every element occurs in exactly one member of $C'$ ?}

\noindent
The following definition gives the construction to reduce \textsc{RX3C} to \DGP.
%We first describe a polynomial-time reduction from RX3C to \DGP~and then prove a lemma allowing us to prove the NP-completeness of \DGP~on cubic graphs. 

\begin{figure}
\begin{minipage}[c]{0.49\linewidth}
    \centering
   \begin{tikzpicture}[thick,scale=0.35]
   \tikzset{node/.style 2 args={draw, circle,draw=black,scale=0.7, label=#1:{\large #2}}};

\node[node={90}{$v_x$},fill ] (1) at (0,0) {};
\node[node={90}{$v_{xy_1z_1}^x$} ] (2) at (3,2) {} edge (1);
\node[node={270}{$v_{xy_2z_2}^x$} ] (3) at (0,-3) {} edge (1);
\node[node={90}{$v_{xy_3z_3}^x$} ] (4) at (-3,2) {} edge (1) ;

\end{tikzpicture}
\caption{Subgraph containing one vertex of type 1, $v_x$, and its neighbors in $G$}
\label{t1}
\end{minipage}
\begin{minipage}[c]{0.49\linewidth}
    \centering
   \begin{tikzpicture}[thick,scale=0.35]
   \tikzset{node/.style 2 args={draw, circle,draw=black,scale=0.7, label=#1:{\large #2}}};

\node[node={90}{$v_{xyz}^x$} ] (1) at (0,0) {};
\node[node={90}{$v_{xyz}^y$} ] (2) at (3,2) {} edge (1);
\node[node={270}{$v_x$},fill ] (3) at (0,-3) {} edge (1);
\node[node={90}{$v_{xyz}^z$} ] (4) at (-3,2) {} edge (1) edge(2) ;

\end{tikzpicture}
\caption{Subgraph containing one vertex of type 1, $v_x$, and three of type 2}
\label{t2}
\end{minipage}

\end{figure}

\begin{definition}\label{construction}
Let $I=(X,C)$ be an instance of $RX3C$. We define the construction $\sigma$ transforming the instance $I$ into the graph $G:= \sigma(I)$ where  $G=(V,E)$  is build as follows (see Figures~\ref{t1} and \ref{t2}): 

\begin{itemize}
    \item for each element $x \in X$, add the vertex $v_x$ to $V$ (called vertices  of type 1 or black vertices).
    \item for each subset of the collection $\{x,y,z\} \in C$, add the vertices $v_{xyz}^x$, $v_{xyz}^y$, $v_{xyz}^z$ to $V$ (called vertices  of type 2 or white vertices).
    \item  add the edges $\{v_{xyz}^x,v_{xyz}^y\}$, $\{v_{xyz}^x,v_{xyz}^z\}$ and $\{v_{xyz}^y,v_{xyz}^z\}$ to $E$  
    \item  add the edges $\{v_{xyz}^x,v_x\}$, $\{v_{xyz}^y,v_y\}$ and $\{v_{xyz}^z,v_z\}$ to $E$ 
\end{itemize}

\noindent
\normalfont{Notice that $G$ is a cubic graph on $|X|$ vertices of type 1 and $3|X|$ vertices of type 2. 
%and contains $6|X|$ edges.
}
\end{definition}

\noindent
Case distinction on the subgraphs in $\sigma(I)$ shows:

\begin{lemma}\label{casescubic}
For $G=(V,E)=\sigma(I)$ and any $P\subseteq V$, it holds that $u(P) \geq \frac{1}{4}$ if and only if $G[P]$ is isomorphic to one of the following three graphs:
\begin{itemize}
    \item a triangle where all the vertices are of type 2 and then $u(P) = \frac{1}{3}$.
    \item an edge between two  type 2 vertices or between  two vertices of different types and then  $u(P) =  \frac{1}{4}$.
    \item the subgraph described in Figure \ref{t2} and then $u(P) =  \frac{1}{4}$.
\end{itemize} 
\label{lem:utiliteG}
\end{lemma}

\begin{proof}
Let $P\subseteq V$ such that $u(P) \geq \frac{1}{4}$. We show in the following that there are exactly three possible  subgraphs $G[P]$ such that  $u(P) \geq \frac{1}{4}$. $G$ obviously does not contain a connected component that is a $K_4$. Also, observe that by its construction, $G$ does not contain $C_4$ as subgraph, since there are no two vertices $u,v\in V$ that have more than one common neighbor. Note that this also implies that $G$ is diamond-free.

As $G$ is cubic, $|E(G[P])| \leq \frac{3}{2}|P|$ and  so $d(P) \leq \frac{3}{2}|P|\cdot \frac{1}{|P|} = \frac{3}{2}$. Since $\frac{1}{4} \leq u(P) \leq \frac{3}{2|P|}$  then $|P| \leq 6$. We study the five following cases:
\begin{itemize}
    \item Case $|P| = 6$: Since $u(P)= \frac{|E(P)|}{6^2}  \geq \frac{1}{4}$,  we have $|E(P)| \geq 9$. Since $G[P]$ cannot be cubic ($G$ is connected and $|V| > 6$), a subgraph with $|P|=6$ and $|E(P)| \geq 9$ does not exist.
    \item Case $|P| = 5$: Since $u(P)= \frac{|E(P)|}{5^2} \geq \frac{1}{4}$,  we have $|E(P)| \geq 7$. Since $G$ contains no $K_4$, the only possibility for this is the graph displayed as Case 1 in Figure~\ref{lem:UtilDiamCubique}. Since $G$ is also diamond-free, such a  subgraph does not exist.
    \item Case $|P| = 4$: Since $u(P)= \frac{|E(P)|}{4^2} \geq \frac{1}{4}$,  we have $|E(P)| \geq 4$.    Since $G$ does not contain a $C_4$ the only possibility for $G[P]$ is the subgraph described in Figure \ref{t2}.
    \item Case $|P| = 3$: Since $u(P) = \frac{|E(P)|}{3^2} \geq \frac{1}{4}$,  we have $|E(P)| \geq 3$ and thus $P$ is a triangle where all the vertices are of type 2 and $u(P) = \frac{1}{3}$.
    \item Case $|P| = 2$: Since $u(P) = \frac{|E(P)|}{2^2}  \geq  \frac{1}{4}$, we have $|E(P)| \geq 1$ and thus $S$ is an edge   between two  type 2 vertices or between two vertices of different types and $u(P) = \frac{1}{4}$.
\end{itemize}
\end{proof}

\begin{remark}
The case-analysis in the proof of \autoref{casescubic} also shows that for any subset $P\subseteq V$ of the vertices of the graph $\sigma(I)$, if $v$ is of type 2 then $u_S(v) \leq \frac{1}{3}$, otherwise $u_S(v) \leq \frac{1}{4}$.
\end{remark}

 With these observations about the construction of $\sigma(I)$, we are able to prove our NP-completeness result.

\begin{theorem}\label{cubichardness}
\DGP~is NP-complete on cubic graphs.
\end{theorem}

\begin{proof}
Let $I=(X,C)$ be an instance of RX3C.
%and consider the following instance $I'$ of \DGP{} on the graph . 
We claim that $I=(X,C)$ is a yes-instance of RX3C if and only if $I'=(G,d)$ with $G=\sigma(I)$ and $d=\frac{7|X|}{6}$ is a yes-instance of \DGP. 

Let $C' \subseteq C$ be an exact cover for $X$ of size $\frac{|X|}{3}$. Consider the following partition $\mathcal{P}$ with $\frac{5|X|}{3}$ parts: for any $c \in C'$, $c=\{x,y,z\}$, we define three parts of size 2, $\{v_x, v_{xyz}^x\}$, $\{v_y, v_{xyz}^y\}$, $\{v_z, v_{xyz}^z\}$ and for any $c \notin C'$, $c=\{x,y,z\}$, we define the following part of size 3,  $\{v_{xyz}^x, v_{xyz}^y, v_{xyz}^z\}$. Since $C'$ is an exact cover, $\mathcal{P}$ is a partition for $G$ and its density is $\frac{3}{2}\cdot\frac{|X|}{3} + \frac{2}{3}|X| = \frac{7}{6}|X|$.

Let $\mathcal{P'}$ be a partition of $G$ of density $d(\mathcal{P'})=\frac{7}{6}|X|$. Firstly, we show  that $\mathcal{P'}$ has necessarily the following shape: $\frac{2|X|}{3}$ parts of size 3 containing only vertices of type 2 forming a triangle in $G$ and $|X|$ parts of size 2 containing one vertex of type 1 and one of type 2 adjacent in $G$ (see Figures \ref{t1} and \ref{t2}). From Remark~\ref{remarkconnexe}, we can assume that all parts induce connected subgraphs.  

We  first show that  $d(\mathcal{P'})=\frac{7|X|}{6}$ implies that there are at least $\frac{2|X|}{3}$ parts in $\mathcal{P'}$ corresponding to triangles in $G$.  
Assume by contradiction that $\mathcal{P'}$ has $\frac{2|X|}{3} - \ell $ triangles, with $\ell > 0$. Since $G$ has $4|X|$ vertices, there are $2|X|+ 3 \ell$ vertices that do not belong to a part in $\mathcal{P'}$ that corresponds to a triangle in $G$.  By Lemma \ref{lem:utiliteG} the utility of these last vertices  is smaller than or equal to $\frac{1}{4}$. Then the density of $\mathcal{P'}$ is 
\[
    d(\mathcal{P}') \leq   \frac{2|X|}{3} - \ell + (2|X|+ 3 \ell) \cdot\frac{1}{4}  = \frac{7|X|}{6} - \frac{\ell}{4}  < \frac{7|X|}{6}
\]
This contradicts the choice of  $\mathcal{P'}$ such that   $d(\mathcal{P'})=\frac{7|X|}{6}$, hence there are at least $\frac{2|X|}{3}$ triangles in $\mathcal{P'}$.

Now, we will prove that there are at most $\frac{2|X|}{3}$ parts in $\mathcal{P'}$ corresponding to triangles in $G$. Assume by contradiction that $\mathcal{P'}$ has $\frac{2|X|}{3} + \ell$ triangles, with $\ell > 0$. Since there are $3|X|$ vertices of type 2 and among these vertices $3\cdot(\frac{2|X|}{3} + \ell)$   belong to a triangle then $|X|-3 \ell$ vertices of type 2 do not belong to a triangle. Each neighbor of a vertex $v_x$ of type 1 is of type 2, so if  the utility of $v_x$ is positive, then there exists a vertex of type 2, $v^x_{xyz}$, neighbor of $v_x$,  that is in the same part as $v_x$ and $v^x_{xyz}$  does not belong to a triangle. Moreover, as all type 1 vertices have no common neighbors, for each type 1 vertex with positive utility, there is a type 2 vertex that is not in a triangle. Since there are at most $|X|-3 \ell$ type 2 vertices that do not belong to a triangle, there are at most $|X|-3 \ell$ type 1 vertices with positive utility.  Then the density of $\mathcal{P'}$ is at most

\[
    d(\mathcal{P'}) \leq \frac{2|X|}{3}+ \ell + \frac{|X| - 3 \ell}{4} + \frac{|X| - 3 \ell}{4} \leq \frac{7|X|}{6} - \frac{\ell}{2} < \frac{7|X|}{6}
\]
This contradicts   the choice of  $\mathcal{P'}$ such that  $d(\mathcal{P'})=\frac{7|X|}{6}$,  and then there are exactly $\frac{2|X|}{3}$ triangles in $\mathcal{P'}$.

We will show now that  $d(\mathcal{P'}) = \frac{7|X|}{6}$ implies that all type 1 vertices are in a part that is a matching with a type 2 vertex. There are $|X|$ type 1 vertices and $|X|$ type 2 vertices that are not in some triangle in $\mathcal{P'}$.  Since there are exactly  $\frac{2|X|}{3}$ parts in $\mathcal{P'}$ forming a triangle and the utility of each other vertex is smaller than or equal to $\frac{1}{4}$, to reach a density of $\frac{7|X|}{6}$ it is necessary that each of the $2|X|$ vertices outside the parts that are triangles  has a utility of exactly $\frac{1}{4}$.  To reach this utility, by Lemma \ref{lem:utiliteG} there are two possibilities, the graph described in Figure \ref{t2} and an edge. Since there are exactly $|X|$ vertices of type 1  and $|X|$  vertices of type 2  outside the triangles in $\mathcal{P'}$, and vertices of type 1 only have neighbors of type 2, the only possibility for all these vertices to have utility $\frac{1}{4}$ is if each type 1 vertex is matched with one type 2 vertex.
%For the first, there are three vertices of type 2 for one of type 1 so we cannot reach $\frac{7|X|}{6}$ so all type 1 vertices are in a matching with a type 2 vertex.  Finally, observe that for every $\{x,y,z\}$ in $C$, if $v^x_{xyz}$ is in a matching with $v_x$, then $v^y_{xyz}$ is in a matching with $v_y$ and  $v^z_{xyz}$ is in a matching with $v_z$. Otherwise, $v^y_{xyz}$ is   in a matching with $v^z_{xyz}$, which is impossible since there are $|X|$ vertices type 1  and $|X|$  vertices type 2  outside the triangles in $\mathcal{P'}$. 

Consider now the following subcollection $C''\subseteq C$: for each triple $v^x_{xyz}$,$v^y_{xyz}$,$v^z_{xyz}$ that does not belong to a triangle, we add the set $\{x,y,z\}$ to $C''$. The subcollection $C''$ is a cover since each type 1 vertex is a neighbor of one of these vertices and it is an exact cover  since there are exactly  $\frac{|X|}{3}$ 3-element subsets that do not belong to a triangle.
\end{proof}

Our observations about the maximum utility of certain vertices can also be used to show the following positive result.

\begin{theorem}
\MDGP{} is polynomial-time $\frac{4}{3}$-approximable  on cubic graphs.
\end{theorem}
\begin{proof}
	Let $I=G$ be a cubic graph, instance of \MDGP{}. If $G$ contains connected components isomorphic to $K_4$, create a part for each such component, as this is the optimum way to partition these sets. So assume that $G$ contains no connected component isomorphic to $K_4$, and let~$D$ be the set of all diamonds in $G$, and $T$ the set of all triangles  that do not belong to a diamond.
	Diamonds (resp. triangles) can be found in polynomial time simply by enumerating all 4-tuples (resp. 3-tuples) of vertices and checking if they induce a diamond (resp. triangle) as subgraph. 
	Let $G'$   be the graph obtained from~$G$ after removing the vertices of $D$ and $T$. Let $M$ be the set of edges  that constitute a maximum matching of $G'$. Let $G''$ be the graph obtained from $G'$ after removing the vertices of $M$. Since $M$ is a maximal matching, the vertices in $G''$ form an  independent set. 
	\medskip
	
	We show in the following that  $|V(G'')| \leq \frac{|V(G)|}{4}$.
	
	For each $v\in V$ we associate a function $t(v)$ and initialize it with $t(v) = 1$. When removing the diamonds and triangles from $G$ in order to get $G'$ we update the function $t$ as follows:
	
	\begin{itemize}
	 \item For every diamond $\{u_1,u_2,u_3,u_4\} \subseteq V$ that is deleted from $V$, let $u_1$ and $u_3$ be the vertices with neighbors outside of the diamond  (if these vertices still exist) and let $v_1$ and $v_3$ be these neighbors (with the possibility that $v_1=v_3$).  We update the function $t$ : $t(v_1) := t(v_1) + t(u_1) + t(u_2)$ and $t(v_3) := t(v_3) + t(u_3) + t(u_4)$ (thus $t(v_1) := t(v_1) + t(u_1) + t(u_2)+ t(u_3) + t(u_4)$ if $v_1=v_3$). If $v_1$ or $v_3$ were already deleted, we delete their associated $t$ function.
	 \item For every triangle $\{u_1,u_2,u_3\} \subseteq V$ that is deleted from $V$, let $v_1$ (resp. $v_2$ and $v_3$) be the neighbor of $u_1$ (resp. $u_2$ and $u_3$) outside of the triangle (if these vertices exist). We update the function $t$ : $t(v_1) := t(v_1) + t(u_1)$, $t(v_2) := t(v_2) + t(u_2)$ and $t(v_3) = t(v_3) + t(u_3)$. If $v_1$, $v_2$ or $v_3$ do not exist, we delete their associated $t$ function.
      \end{itemize}

Observe that after  updating  $t$ for any $v \in V(G')$, if $v \in D_{G'}(3)$ then $t(v) \geq 1$, if $v \in D_{G'}(2)$ then $t(v) \geq 2$, if $v \in D_{G'}(1)$ then $t(v) \geq 3$ and if $v \in D_{G'}(0)$ then $t(v) \geq 4$.
 In order to justify this, observe that the $t$ function associated to vertices in $V(G')$ cannot decrease. If a vertex $v$ is of degree $3-i$  in $G'$, $1\leq i\leq 3$, then there are at least $i$ adjacent edges to distinct vertices in triangles or diamonds that were removed from $G$ and increase $t(v)$. Each time when a neighbor of  a vertex $v$  from a diamond or a triangle is removed then $t(v)$ increases by at least one. Then, in $G'$, each vertex $v$ of degree $3-i$ has $t(v)\geq i+1$.

Let $n'_i$ be the number of vertices of degree $i$ in $G'$. By the previous remark, we have 
\begin{equation}
\sum\limits_{v \in V(G')} t(v_i) \geq 4n'_0+ 3n'_1 + 2n'_2 + n'_3
\label{lemmedegree}
\end{equation}

 Since $G'$ is a subcubic triangle-free graph and  $M$  a maximum matching in $G'$, using a result of Munaro \cite{munaro2017line}, we get  
 \begin{equation}
 |V(M)| \geq \frac{9}{10}n'_3 + \frac{3}{5}n'_2 + \frac{3}{10}n'_1
 \label{lemmemunaro}
\end{equation}
 
We show now that  $4|V(G'')| \leq \sum\limits_{v \in V(G')} t(v_i)$. In fact, combining $|V(G')|=  n'_0+ n'_1 + n'_2 + n'_3$  with inequality~(\ref{lemmemunaro}) gives $|V(G'')| \leq n'_0+ \frac{7}{10}n'_1 + \frac{2}{5}n'_2 + \frac{1}{10}n'_3$.  
Thus, $4|V(G'')| \leq 4n'_0 + \frac{28}{10}n'_1 + \frac{8}{5}n'_2 + \frac{4}{10}n'_3 \leq 4n'_0 + 3n'_1 + 2n'_2 + n'_3\leq \sum\limits_{v \in V(G')} t(v_i)$ using  inequality~(\ref{lemmedegree}).  Then 
 $4|V(G'')| \leq \sum\limits_{v \in V(G')} t(v_i)$ and
  since $|V(G)| \geq \sum\limits_{v \in V(G')} t(v_i)$ we get $|V(G'')| \leq \frac{1}{4} V(G)$.

		\medskip
	 Consider the partition $\mathcal{P} = D \cup T \cup M \cup V(G'')$ in the sense that $\mathcal{P}$ contains a set for each diamond in $D$, one set for each triangle in $T$, one set for each edge in the matching $M$ and one set for each vertex in $V(G'')$.
	  Then $d(\mathcal{P}) = \frac{5}{4}|D| + |T| + \frac{1}{2} |M| \geq \frac{5}{4}|D| + |T| + \frac{1}{4}(n - 3|T| - 4|D| - \frac{n}{4})$ since $|V(G'')| \leq \frac{1}{4} V(G)$. By \autoref{lem:ApproxBorneCubique} we know that $opt(I) \leq \frac{5}{4}|D| + |T| + \frac{1}{4}(n - 3|T| - 4|D|)$. Then $\frac{opt(I)}{d(\mathcal{P})} \leq \frac{\frac{5}{4}|D| + |T| + \frac{1}{4}(n - 3|T| - 4|D|)}{\frac{5}{4}|D| + |T| + \frac{1}{4}(n - 3|T| - 4|D| - \frac{n}{4})} = \frac{\frac{1}{4}|D| + \frac{1}{4}|T| + \frac{n}{4}}{\frac{1}{4}|D| + \frac{1}{4}|T| + \frac{3n}{16}} = 1+\frac{{n}}{4|D|+4|T|+{3n}}\leq 1+ \frac{1}{3}$. 
	  Then $\frac{opt(I)}{d(\mathcal{P})} \leq \frac{4}{3}$.
\end{proof}

 \section{Dense Graphs}\label{sec5}
 \label{sec:densGrapheDensite}
 In this section we consider graphs $G=(V,E)$ on $n$ vertices such that $G$ can be viewed as $G=\overline{H}$ where $H$ is a graph of small maximum degree. Note that the edges of $H$ are exactly the \emph{missing edges} of $G$. We first consider graphs $G=(V,E)$ on $n$ vertices such that $\delta(G) \geq n-3$, that is  $G=\overline{H}$ where $H$ has $\Delta(H)=2$  and has $q\leq n$ edges and show that \MDGP{} is solvable in polynomial time on these graphs. 
 
 \begin{lemma}\label{threesets}
For any graph $G$ on $n$ vertices such that $\delta(G) \geq n-3$, its density $d(G)$ is greater than or equal to the density of any partition $\mathcal{P}$ of $G$ into $t\geq 3$ parts.
 \end{lemma}
 \begin{proof} 
 The density of $G$ is given by  $d(G)= \frac{\frac{n(n-1)}{2}-q}{n}=\frac{n-1}{2} - \frac{q}{n}$. From Lemma~\ref{lemmacomplete}, among all partitions of $G$ into $t\geq 3$ parts, those where the parts correspond to complete graphs have the largest density.  The density of such a partition into $t$ parts of size $n_1,\ldots,n_t$ is $\frac{n-t}{2}$.  Thus, the density of $G$ is at least as large as the density of this last partition since $t\geq 3$ and $q\leq n$ (note here that a graph with minimum degree $n-3$ has at most $n$ missing edges).
 \end{proof}

Observe that in the proof of the previous lemma when $q=n$ and $t=3$, the density of a partition in 3 parts corresponding to complete subgraphs and the density of the entire graph are the same.    This previous lemma implies that for any graph $G$  such that $\delta(G) \geq n-3$, there exists a partition into one or two parts of maximum density.

\begin{lemma}\label{odd}
For any graph $G$ on $n$ vertices such that $\delta(G) \geq n-3$, in any partition for $G$ into two parts, the sum of missing edges in the two parts is at least $o$, where $o$ is the number of odd cycles in $\overline{G}$.
\end{lemma}

\begin{proof}
Let $C$ be an odd cycle in $\overline{G}$ (the graph of missing edges in $G$). Since $C$ is not bipartite, there is no partition $\{V_1,V_2\}$ of $V$ such that all the edges of $C$ have one endpoint in $V_1$ and one endpoint in $V_2$. Hence, for any partition $\{V_1,V_2\}$ at least one of the missing edges from $C$ is  inside $G[V_1]\cup G[V_2]$. 
\end{proof} 
 
 \begin{lemma}\label{part1}
 Among all partitions into 2 parts of fixed size containing $x$ missing edges, the one containing all missing edges in the largest part has the best density. 
 \end{lemma}

 \begin{proof}
 Consider two partitions $\{V_1,V_2\}$ and $\{V_1',V_2'\}$ such that $|V_1|=|V_1'|=n_1$ and $|V_2|=|V_2'|=n_2$ with $n_1\leq n_2$ and $G[V_1]$ (resp. $G[V_2]$) containing $x_1$ (resp. $x_2$) missing edges and $G[V_1']$ (resp. $G[V_2']$) containing $0$ (resp. $x=x_1+x_2$) missing edges. The densities for these partitions are:
 \begin{itemize}
 \item[] $d(\{V_1,V_2\})= \frac{n-2}{2}- \frac{x_1}{n_1} - \frac{x_2}{n_2}$, and
 \item[]  $d(\{V_1',V_2'\})= \frac{n-2}{2}- \frac{x}{n_2}$.
  \end{itemize}
  Since $x=x_1+x_2$ and $n_1\leq n_2$, it follows that $d(\{V_1,V_2\}) \leq d(\{V_1',V_2'\})$. 
 \end{proof}

 \begin{lemma}\label{part2}
 Among all partitions into 2 parts  containing 0 (resp. $x$) missing edges in the smaller (resp. larger) part, the one with a maximum number of vertices in the largest part has the best density. \
 \end{lemma}

\begin{proof}
 Consider two partitions $\{V_1,V_2\}$ and $\{V_1',V_2'\}$ such that $|V_1|=n_1$,  $|V_2|=n_2$ with $n_1\leq n_2$ and $|V_1'|=n_1'$,  $|V_2'|=n_2'$ with $n_1'\leq n_2'$ and $G[V_1]$ (resp. $G[V_2]$) containing $0$ (resp.~$x$) missing edges and $G[V_1']$ (resp. $G[V_2']$) containing $0$ (resp.~$x$) missing edges. Moreover suppose $n_2\leq n_2'$. The densities for these partitions are:
 \begin{itemize}
 \item[] $d(\{V_1,V_2\})= \frac{n-2}{2}-  \frac{x}{n_2}$, and
 \item[]  $d(\{V_1',V_2'\})= \frac{n-2}{2}- \frac{x}{n_2'}$.
 \end{itemize} 
  Since  $n_2\leq n_2'$, it follows that $d(\{V_1,V_2\}) \leq d(\{V_1',V_2'\})$. 
\end{proof}
 
 \begin{theorem}\label{polydense}
 \MDGP{} is solvable in polynomial time on graphs~$G$ with $n$ vertices with $\delta(G)\geq n-3$.
 \end{theorem}
 
 \begin{figure}
    \dispFig{
        \centering
         \begin{tikzpicture}[scale = 0.8,thick]
          \usetikzlibrary{patterns}
    		\tikzset{node/.style 2 args={draw, circle,draw=black,scale=0.9, label=#1:{\large #2}}};
    		\node[node={270}{},pattern=north east lines ] (0) at (0,0) {}; 
    		\node[node={270}{},pattern=north east lines ] (1) at (1,0) {}; 
    		\node[node={270}{},pattern = north east lines ] (2) at (2,0) {};
    		\node[node={270}{},pattern=north east lines ] (3) at (3,0) {};
    		\node[node={270}{},pattern=north east lines ] (4) at (4,0) {};
    		\node[node={270}{},pattern=north east lines ] (5) at (5,0) {};
    		\node[node={270}{},pattern=north east lines ] (6) at (6,0) {};
    		\node (a0) at (-1,0) {$V_2$};
    		\node[node={270}{},pattern=dots ] (7) at (1,2) {};
    		\node[node={270}{},pattern=dots ] (8) at (2,2) {};
    		\node[node={270}{},pattern=dots] (9) at (3,2) {};
    		\node[node={270}{} ,pattern=dots] (10) at (4,2) {};
    		\node[node={270}{},pattern=dots] (11) at (5,2) {};
    		\node (a7) at (0,2) {$V_1$};
    		
    		\foreach \j in {0,...,6}{
    			\foreach \i in {0,...,6}{
    				\ifthenelse{\i < \j}{}{
    					\draw (\i) edge[gray!60, bend left = 15] (\j);
    				}
    			}
    		}
    		
    		\foreach \j in {7,...,11}{
    			\foreach \i in {7,...,11}{
    				\ifthenelse{\i < \j}{}{
    					\draw (\i) edge[gray!60, bend right = 15] (\j);
    				}
    			}
    		}
    		\foreach \j in {0,...,6}{
    			\foreach \i in {7,...,11}{
    				\ifthenelse{\i < \j}{}{
    					\draw (\i) edge[gray!60] (\j);
    				}
    			}
    		}
    	
    		\draw (0)--(7);
    		\draw (7)--(1);
    		\draw (1)--(8);
    		\draw (8)--(2);
    		\draw (2) edge[bend left	 = 15] (0);
    		
    		\draw (3)--(9);
    		
    		\draw (4)--(10);
    		\draw (10)--(5);
    		\draw (5)--(11);
    		\draw (11)--(4);
    		
    	    \draw[dashed] \convexpath{a0,6}{0.6 cm};
    	    \draw[dashed] \convexpath{a7,11}{0.6 cm};

     	\end{tikzpicture}
    }
    \caption{Construction of $V_1$ and $V_2$ in Theorem~\ref{polydense}}
    \label{fig:polyDens}
 \end{figure}
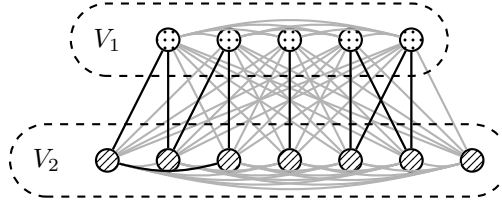
 
 \begin{proof} Let $G$ be a graph of minimum degree $n-3$.  We first define a partition $\{V_1,V_2\}$ of the vertices of $G$ by giving vertices color 1 or~2, in the sense that $V_1$ (resp. $V_2$) contains vertices of color 1 (resp. 2). An example is given in Figure \ref{fig:polyDens}.
We assign color 2 to each vertex of degree $n-1$. Since the minimum degree in $G$ is $n-3$, the graph $H$ of missing edges is a collection of paths and cycles. We color the vertices on paths or cycles with an even number of vertices alternating by 1 and 2. For vertices on paths or cycles with an odd number of vertices we also color them alternating by 1 and 2, always starting with color 2. Thus cycles of odd size have two adjacent vertices of color 2. 

Let $o$ be the number of odd cycles in $H$. The partition $\{V_1,V_2\}$ defined by our 2-coloring contains $o$ missing edges in $V_2$ and $|V_2|$ is maximized among all such partitions.  
 Its density is equal to   $\frac{n-2}{2}- \frac{o}{n_2}$,  where $n_2=|V_2|$. Denote by $d_{n-1}$ the number of vertices of~$G$ of degree $n-1$ and by $p_o$ the number of paths with an odd number of vertices (even length) among the missing edges.  The sets $V_1$ and $V_2$ contain the same number of vertices of degree $n-2$ that are extremities of a path with an even number of vertices in $H$.
 The set $V_2$ contains $p_o$ more vertices  of degree $n-2$,  that are extremities of a path with an odd number of vertices,  than $V_1$. The set $V_2$ contains $o$ more vertices  of degree $n-3$  than $V_1$. Thus $n_1=\frac{1}{2}(n-d_{n-1}-p_o-o)$ and $n_2=\frac{1}{2}(n+d_{n-1}+p_o+o)$.
  We claim that there is no partition into two parts that has a higher density. 
 
 By Lemma~\ref{odd}, any partition into two sets contains at least $o$ missing edges inside the two parts. By construction we have maximized the number of vertices in the part with the missing edges among all partitions with the minimum number $o$ of missing edges, i.e., there is no partition into two parts  $\{V'_1,V'_2\}$ with $o$ missing edges all contained in $V'_2$ and $|V'_2|>|V_2|$. Hence, by Lemmas~\ref{part1} and~\ref{part2},  it remains to show that any partition $\{V'_1,V'_2\}$ with $o+x>o$ missing edges for some $x>0$ has a smaller density than  $\{V_1,V_2\}$. 

Let $\{V'_1,V'_2\}$ be a partition with $o+x>o$ missing edges for some $x> 0$   and assume w.l.o.g.~that $|V'_1|\leq |V'_2|$. By definition of the partition $\{V_1,V_2\}$, it follows that $|E(H)|=2n_1-p_o+o$ (note that all non-edges have to either be among the $o$ missing edges in the partition or in the cut between $V_1$ and $V_2$). In the partition  $\{V'_1,V'_2\}$, it follows that $|E(H)|\leq 2|V'_1|-r_1+(o+x)$, where $r_1$ is the number of vertices in $V_1'$ adjacent to only one edge in $H$. In the cut between $V'_1$ and $V'_2$, each vertex in $V'_1$ is adjacent to at most two such edges. Combining these two bounds on $|E(H)|$ yields 
  \begin{equation}\label{dese_eq1}
 2n_1-p_o\leq 2|V'_1|-r_1+x\,.
\end{equation} 
  We claim that $r_1\geq p_o-x$. To see this, observe that every path of odd length either results in a vertex in $V'_1$ adjacent to only one edge  in $E(H)$ ($r_1$) or in a missing edge. Also, every cycle of odd length creates at least one missing edge. Thus the number of missing edges $o+x$ for  $\{V_1',V_2'\}$   is at least $p_o-r_1+o$. Reordering this yields the claimed 
  \begin{equation}\label{dese_eq2}
 r_1\geq p_o-x\,.
\end{equation} 
 
 Inequalities~(\ref{dese_eq1}) and  (\ref{dese_eq2}) yield $2n_1-p_o\leq 2|V'_1|-p_o+2x$ and thus $n_1-|V'_1|\leq x$. Since $\{V'_1,V'_2\}$ is a partition it follows that $|V'_2|=n-|V'_1|\leq n-n_1+x=n_2+x$.

By Lemmas~\ref{part1} and~\ref{part2}, the best case of missing edges for $\{V'_1,V'_2\}$ is that they all are in the larger part $V'_2$, hence the density of $\{V'_1,V'_2\}$ is at most $\frac{n-2}{2}- \frac{o+x}{|V'_2|}$. With $|V'_2|\leq n_2+x$, we can bound  $d(V'_1,V'_2)\leq \frac{n-2}{2}- \frac{o+x}{n_2+x}$. Since $H$ is of degree at most 2, we know that there cannot be more missing edges than vertices in a part, thus in particular $o\leq {n_2}$. This last observation allows to bound $d(V'_1,V'_2)\leq \frac{n-2}{2}- \frac{o+x}{n_2+x}\leq \frac{n-2}{2}- \frac{o}{n_2}=d(V_1,V_2)$, thus the density of $\{V'_1,V'_2\}$ is not larger than the density of $\{V_1,V_2\}$. 
 \end{proof}
 
In the rest of the section  we consider graphs $G=(V,E)$ on $n$ vertices, $(n-4)$-regular, that is  $G=\overline{H}$ where $H$ is a cubic graph. 
We show that \DGP~is NP-hard on $(n-4)$-regular graphs, by showing  a reduction from \UC{} on cubic graphs,  that is the complement of \textsc{Max Cut}.  This last problem on cubic graphs was proved  NP-hard and even not polynomial-time  1.003-approximable, unless P=NP \cite{BK99}.

\defprob{Min UnCut}{A graph $G=(V,E)$, an integer $k$.}{Does $G$ contain a partition of V into two parts $A, B$ such that the number of edges with both endpoints in the same part is at most $k$?}

Since we reduce from \UC{} on cubic graphs, we use the following straightforward observation on any partition in such graphs.
\begin{lemma} \label{lem:interneexterne}For any cubic graph $G$ and any $\{A,B\}$   partition of $V$, we have 
$ |A| + \frac{2}{3}\cdot |E(B)| = |B| + \frac{2}{3} \cdot |E(A)|$, 
where $E(A)$, resp.~$E(B)$, is the set of edges with both endpoints in $A$, resp~$B$.
\end{lemma}

Since we did not find a reference for the following result in the literature we propose a short proof. 

\begin{lemma}\label{perhapsknown}
\label{lem:borneUnCut}
Let $G=(V,E)$ be a cubic graph. There exists a partition $\{A,B\}$ of $G$ with a cut of size at least $|V|$ and it can be found in polynomial time.
\end{lemma}

\begin{proof}
Let $\mathcal{P} = \{A,B\}$ be a partition of $V$. Consider the following operation: if there is a vertex $v \in A$ (resp. $B$) with at least two neighbors in $A$ (resp. $B$) then $A = A \setminus \{v\} $ (resp. $B = B \setminus \{v\} $) and $B =  B \cup \{v\} $ (resp. $A = A \cup \{v\} $). Since the graph is cubic, this operation increases the number of edges between $A$ and $B$ by at least one. Since the number of edges is finite, we can repeat this operation until we obtain a partition $P'=\{A',B'\}$ with no vertex $v \in A'$ (resp. $B'$) with at least two neighbors in $A'$ (resp. $B'$). Since the graph is cubic, if every vertex in $A'$ (resp. $B'$) has at most one neighbor in $A'$, then it has at least two neighbors in $B'$ (resp. $A'$). Consequently $P'$ has a cut of size at least $ \frac{2(|A'| + |B'|)}{2} = |V|$.
\end{proof}

\begin{figure}
\begin{center}
    \begin{tikzpicture}[scale = 0.6,thick]
	\tikzset{node/.style 2 args={draw, circle,draw=black,scale=0.9, label=#1:{\large #2}}};
		\node[node={270}{} ] (0) at (0,0) {}; 
		\node[node={270}{} ] (1) at (1,0) {}; 
		\node[node={270}{} ] (2) at (2,0) {};

		\node[node={270}{} ] (3) at (6,0) {};
		\node[node={270}{} ] (4) at (7,0) {};
		\node[node={270}{} ] (5) at (8,0) {};

		\node[node={270}{} ] (7) at (0,3) {};
		\node[node={270}{} ] (8) at (1,3) {};
		\node[node={270}{} ] (9) at (2,3) {};

		\node[node={270}{} ] (10) at (6,3) {};
		\node[node={270}{} ] (11) at (7,3) {};
		\node[node={270}{} ] (12) at (8,3) {};
		
		\foreach \j in {0,...,5}{
			\foreach \i in {7,...,12}{
				\ifthenelse{\i < \j}{}{
					\draw (\i) edge[gray!30] (\j);
				}
			}
		}
		\draw (0) edge[gray!30, bend right = 15] (2);
		\draw (0) edge[gray!30, bend right = 15] (3);
		\draw (0) edge[gray!30, bend right = 15] (4);
		\draw (0) edge[gray!30, bend right = 15] (5);
		\draw (1) edge[gray!30] (2);
		\draw (1) edge[gray!30, bend right = 15] (3);
		\draw (1) edge[gray!30, bend right = 15] (4);
		\draw (1) edge[gray!30, bend right = 15] (5);
		
		\draw (2) edge[gray!30, bend right = 15] (4);
		\draw (2) edge[gray!30, bend right = 15] (5);
		
		\draw (3) edge[gray!30] (4);
		\draw (3) edge[gray!30, bend right = 15] (5);
		
		\draw (12) edge[gray!30, bend right = 15] (10);
		\draw (12) edge[gray!30, bend right = 15] (9);
		\draw (12) edge[gray!30, bend right = 15] (8);
		\draw (12) edge[gray!30, bend right = 15] (7);
		\draw (11) edge[gray!30] (10);
		\draw (11) edge[gray!30, bend right = 15] (9);
		\draw (11) edge[gray!30, bend right = 15] (8);
		\draw (11) edge[gray!30, bend right = 15] (7);
		
		\draw (10) edge[gray!30, bend right = 15] (8);
		\draw (10) edge[gray!30, bend right = 15] (7);
		
		\draw (9) edge[gray!30] (8);
		\draw (9) edge[gray!30, bend right = 15] (7);
		
		\draw (0) edge[gray!30] (1);
		\draw (12) edge[gray!30] (11);
		\draw (8) edge[gray!30] (7);
		\draw (4) edge[gray!30] (5);

		\draw[dotted, line width = 1.5] (2)--(3);
		\draw[dotted, line width = 1.5] (9)--(10);

		%-------------------------------------------------------

		\node[node={270}{} ] (a0) at (17,0) {}; 
		\node[node={270}{} ] (a2) at (13,0) {};
		\node[node={270}{} ] (a3) at (14,0) {};
		\node[node={270}{} ] (a4) at (10,0) {};
		\node[node={270}{} ] (a6) at (16,0) {};
		\node[node={270}{} ] (a7) at (15,0) {};
		\node[node={270}{} ] (a5) at (12,0) {};
		\node[node={270}{} ] (a1) at (11,0) {};
		
		\node[node={270}{} ] (a9) at (12,3) {};
		\node[node={270}{} ] (a10) at (13,3) {};
		\node[node={270}{} ] (a11) at (14,3) {};
		\node[node={270}{} ] (a12) at (15,3) {};
		
		\foreach \j in {0,...,7}{
			\foreach \i in {9,...,12}{
				\ifthenelse{\i < \j}{}{
					\draw (a\i) edge[gray!30] (a\j);
				}
			}
		}
		
		\foreach \i in {9,...,12}{
			\foreach \j in {9,...,12}{
				\ifthenelse{\i < \j}{}{
					\draw (a\i) edge[gray!30, bend right = 15] (a\j);
				}
			}
		}
		
		\draw (a4) edge[gray!30, bend right = 15] (a0);
		\draw (a4) edge[gray!30, bend right = 15] (a2);
		\draw (a4) edge[gray!30, bend right = 15] (a3);
		\draw (a4) edge[gray!30, bend right = 15] (a5);
		\draw (a4) edge[gray!30, bend right = 15] (a6);
		\draw (a4) edge[gray!30, bend right = 15] (a7);
		
		\draw (a1) edge[gray!30] (a5);
		\draw (a1) edge[gray!30, bend right = 15] (a0);
		\draw (a1) edge[gray!30, bend right = 15] (a2);
		\draw (a1) edge[gray!30, bend right = 15] (a3);
		\draw (a1) edge[gray!30, bend right = 15] (a6);
		\draw (a1) edge[gray!30, bend right = 15] (a7);
		
		\draw (a5) edge[gray!30, bend right = 15] (a0);
		\draw (a5) edge[gray!30, bend right = 15] (a3);
		\draw (a5) edge[gray!30, bend right = 15] (a6);
		\draw (a5) edge[gray!30, bend right = 15] (a7);
		
		\draw (a2) edge[gray!30, bend right = 15] (a0);
		\draw (a2) edge[gray!30, bend right = 15] (a6);
		\draw (a2) edge[gray!30, bend right = 15] (a7);
		
		\draw (a3) edge[gray!30, bend right = 15] (a0);
		\draw (a3) edge[gray!30, bend right = 15] (a6);
		
		\draw (a7) edge[gray!30, bend right = 15] (a0);

		\foreach \i in {9,...,12}{
			\foreach \j in {7,...,12}{
				\ifthenelse{\i < \j}{}{
					\draw (a\i) edge[gray!30, bend right = 10] (\j);
				}
			}
		}
	
		\foreach \i in {0,...,7}{
			\foreach \j in {0,...,5}{
				\ifthenelse{\i < \j}{}{
					\draw (a\i) edge[gray!30, bend left = 9] (\j);
				}
			}
		}
	
		\foreach \i in {9,...,12}{
			\foreach \j in {0,...,5}{
					\draw (a\i) edge[gray!30] (\j);
			}
		}
	
		\foreach \i in {0,...,7}{
			\foreach \j in {7,...,12}{
				\draw (a\i) edge[gray!30] (\j);
			}
		}

		\draw[line width=0.8pt] (a0)--(a12);
		\draw[line width=0.8pt]  (a2)--(a3);
		\draw[line width=0.8pt]  (a4)--(a9);
		\draw[line width=0.8pt]  (a6)--(a7);
		\draw[line width=0.8pt]  (a5)--(a11);
		\draw[line width=0.8pt]  (a10)--(a1);
		\draw[line width=0.8pt]  (a0)--(a6);
		\draw[line width=0.8pt]  (a3)--(a7);
		\draw[line width=0.8pt]  (a12)--(a7);
		\draw[line width=0.8pt]  (a12)--(a5);
		\draw[line width=0.8pt]  (a2)--(a5);
		\draw[line width=0.8pt]  (a2)--(a11);
		\draw[line width=0.8pt]  (a0)--(a11);
		\draw[line width=0.8pt]  (a3)--(a10);
		\draw[line width=0.8pt]  (a4)--(a10);
		\draw[line width=0.8pt]  (a4)--(a1);
		\draw[line width=0.8pt]  (a9)--(a1);
		\draw[line width=0.8pt]  (a9)--(a6);
		
		\foreach \j in {0,...,2}{
			\foreach \i in {7,...,9}{
				\ifthenelse{\i < \j}{}{
					\draw (\i) edge (\j);
				}
			}
		}
		
		\foreach \j in {3,...,5}{
			\foreach \i in {10,...,12}{
				\ifthenelse{\i < \j}{}{
					\draw (\i) edge (\j);
				}
			}
		}
		
		\draw[dashed] \convexpath{0,a0}{0.9 cm};
		\draw[dashed] \convexpath{7,a12}{0.9 cm};

		\node[node={270}{} ] (0) at (0,0) {}; 
		\node[node={270}{} ] (1) at (1,0) {}; 
		\node[node={270}{} ] (2) at (2,0) {};

		\node[node={270}{} ] (3) at (6,0) {};
		\node[node={270}{} ] (4) at (7,0) {};
		\node[node={270}{} ] (5) at (8,0) {};

		\node[node={270}{} ] (7) at (0,3) {};
		\node[node={270}{} ] (8) at (1,3) {};
		\node[node={270}{} ] (9) at (2,3) {};

		\node[node={270}{} ] (10) at (6,3) {};
		\node[node={270}{} ] (11) at (7,3) {};
		\node[node={270}{} ] (12) at (8,3) {};
		
		\node[node={270}{} ] (a0) at (17,0) {}; 
		\node[node={270}{} ] (a2) at (13,0) {};
		\node[node={270}{} ] (a3) at (14,0) {};
		\node[node={270}{} ] (a4) at (10,0) {};
		\node[node={270}{} ] (a6) at (16,0) {};
		\node[node={270}{} ] (a7) at (15,0) {};
		\node[node={270}{} ] (a5) at (12,0) {};
		\node[node={270}{} ] (a1) at (11,0) {};
		
		\node[node={270}{} ] (a9) at (12,3) {};
		\node[node={270}{} ] (a10) at (13,3) {};
		\node[node={270}{} ] (a11) at (14,3) {};
		\node[node={270}{} ] (a12) at (15,3) {};
	\end{tikzpicture}
    \caption{The construction of $G'$ in Definition~\ref{denseNPdef}}
    \label{fig:hardDens}
    \end{center}
\end{figure}
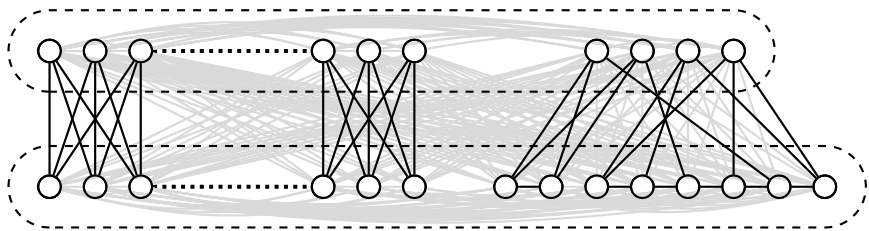

\begin{definition}\label{denseNPdef}
Let $I = (G,k)$ be an instance of \UC~ where $G=(V,E)$ is a cubic graph. We define the construction $\sigma$ transforming the graph $G$ into the graph $G':= (V',E') =\sigma(G)$ (see Figure \ref{fig:hardDens}) as follows:
\begin{itemize}
    \item let $G_0=(V_0,E_0)$ be the union of  $\frac{n^2 - n}{6}$ copies of $K_{3,3}$ (see remark below). Thus $G_0$ is a cubic bipartite graph with $n^2 - n$ vertices and $V_0$ is the union of two independent sets $L, R$ such that $|L|=|R|$.
    \item let $G_1 = (V \cup V_0, E \cup E_0)$.
    \item let $G' = \overline{G_1}$.
\end{itemize}
 
\end{definition}

\begin{remark}
Note that  we can assume that the number of vertices of a cubic graph $G$  is a multiple of $6$. Since $G$ is cubic, $n$ is a multiple of 2.
If $n$ is not a multiple of 3, we consider the instance $I_{triple}$ defined as follows: $G_{triple}$ is the union of 3 copies of $G$ and $k_{triple}=3k$, and thus in the new instance $I_{triple}$ the graph has $3n$ vertices. 
 Note that the number of edges with both endpoints in the same part is
$3k$ in $G_{triple}$ if and only if it is $k$ in $G$.
\end{remark}

Let $n=|V|$, $m=|E|$,  $n'=|V'|$ and $m'=|E'|$. Observe that $n' = n^2$, and  $G'$ is a $(n'-4)$-regular graph. 

\begin{proof}
Since the graph $G$ is cubic, 
$|E(A,B)|  = 3\cdot |A| - 2 \cdot|E(A)|  = 3\cdot |B| - 2 \cdot |E(B)|$.
We can deduce that $|A| + \frac{2}{3} \cdot |E(B)| = |B| + \frac{2}{3} \cdot |E(A)|$ 
\end{proof}

\begin{theorem}\label{denseNP}
\DGP{} is NP-complete on $(n-4)$-regular graphs with $n$ vertices.
\end{theorem}

\begin{proof}
Let $I=(G=(V,E),k)$ be an instance of \UC,  where $G$ is a cubic graph.  Consider the following instance $I'$ of \DGP{} on the graph $G' = \sigma(G)$ and $d=\frac{n^2}{2} - 1 - \frac{2k}{n^2}$. We claim that $I= (G,k)$ is a yes-instance of \UC~ if and only if $I' = (G',d)$ is a yes-instance of \DGP.
\medskip

Let $\{A,B\}$ be a partition of $V$ whose  uncut value is at most $k$. Since $V_0=L\cup R$, where $L,R$ are independent sets in $G_0$ such that $|L|=|R|$, the sets $L,R$ form two cliques of the same size in $G'$. Let $A' = A \cup L$ and $B' = B \cup R$  and $\mathcal{P}=\{A',B'\}$ be a partition of $G'$. 

Let ${M}_{A'}$ and ${M}_{B'}$ be the set of missing edges in $G'[A']$ and $G'[B']$, respectively. 
Due to the construction of $G'$, there is no missing edge  between $A$ and $L$ and between $B$ and $R$. Thus all missing edges are inside $G'[A\cup B]$, i.e. $|{M}_{A'}| + |{M}_{B'}| \leq k$. Thus, the density of the partition $\mathcal{P}$ is: 
$$ d(\mathcal{P}) = \frac{|A'|-1}{2} - \frac{|{M}_{A'}|}{|A'|} + \frac{|B'|-1}{2} - \frac{|{M}_{B'}|}{|B'|} 
       = \frac{n^2-2}{2} - \frac{|{M}_{A'}|}{|A'|} - \frac{|{M}_{B'}|}{|B'|}
$$
We will prove in the following that  $d(\mathcal{P})
\geq d=\frac{n^2}{2} - 1 - \frac{2k}{n^2}$ that is equivalent to proving that $\frac{|{M}_{A'}|}{|A'|} + \frac{|{M}_{B'}|}{|B'|} \leq \frac{2(|M_{A'}|+|M_{B'}|)}{|A'|+|B'|}$.

Consider the difference 
$$  \frac{2(|M_{A'}|+|M_{B'}|)}{|A'|+|B'|} - \left(\frac{|M_{A'}|}{|A'|} + \frac{|M_{B'}|}{|B'|}\right)=$$
$$=\frac{1}{|A'|+|B'|} \left( 2|M_{A'}|+2|M_{B'}| - \frac{|A'|+|B'|}{|A'|} |M_{A'}|-\frac{|A'|+|B'|}{|B'|} |M_{B'}|\right)=$$
$$=\frac{1}{|A'|+|B'|} \frac{1}{|A'|} \frac{1}{|B'|} (|A'||B'| |M_{A'}|+ |A'||B'| |M_{B'}|-|B'|^2 |M_{A'}|- |A'|^2 |M_{B'}|)=
$$
$$ =\frac{1}{|A'|+|B'|} \frac{1}{|A'|} \frac{1}{|B'|} (|A'|-|B'|)(|B'||M_{A'}|-|A'||M_{B'}|)
$$

Wlog we can consider that $|A'|\geq |B'|$, that implies $|B'|\leq \frac{n^2}{2}$.
From Lemma~\ref{lem:interneexterne}  for the cubic graph $G_1$ and partition $\{A', B'\}$,  we have 
$|A'|+ \frac{2}{3}\cdot |M_{B'}| = |B'| + \frac{2}{3} \cdot |M_{A'}|$. Using that $|A'|=n^2-|B'|$ and $|M_{A'}|=k-|M_{B'}|$, we have $n^2-|B'|+ \frac{2}{3}\cdot |M_{B'}| = |B'| + \frac{2}{3} \cdot (k-|M_{B'}|)$  and thus $ |M_{B'}|= \frac{3}{4} (2 |B'|+ \frac{2}{3} k-n^2)$.

Thus, 
$$|B'||M_{A'}|-|A'||M_{B'}|= |B'| (k-|M_{B'}|) - (n^2-|B'|)|M_{B'}|=|B'|k-n^2 |M_{B'}|=$$
$$=|B'|k-n^2\frac{3}{4} \left(2 |B'|+ \frac{2}{3} k-n^2\right)=\left(|B'|-\frac{n^2}{2}\right)\left(k-\frac{3n^2}{2}\right)$$
Since $|B'|\leq \frac{n^2}{2}$ and $k \leq  \frac{n}{2}\leq \frac{3n^2}{2}$ we can conclude  that 
$$  \frac{2(|M_{A'}|+|M_{B'}|)}{|A'|+|B'|} - \left(\frac{|M_{A'}|}{|A'|} + \frac{|M_{B'}|}{|B'|}\right) \geq 0$$

Thus, the partition $\mathcal{P}=\{A',B'\}$  has the density $d(\mathcal{P})
\geq d=\frac{n^2}{2} - 1 - \frac{2k}{n^2}$.

 \bigskip
 Let $\mathcal{P'}$ be a partition of $G'$ of density $d(\mathcal{P'}) \geq d = \frac{n^2-2}{2} - \frac{2k}{n^2}$. We will prove that $\mathcal{P'}$ has exactly two parts $A'$ and $B'$ such that $A = A' \cap V$ and $B = B' \cap V$ is a partition of $G$ whose uncut value is at most $k$.

Suppose that $|\mathcal{P'}| \geq 3$. Then, using Lemma~\ref{lem:densMax}, we have $d(\mathcal{P'}) \leq \frac{n^2 - |\mathcal{P'}|}{2} \leq \frac{n^2-3}{2} = \frac{n^2-2}{2} - \frac{1}{2}$. Since $k \leq \frac{n}{2}$ and $n \geq 6$ then $\frac{2k}{n^2} < \frac{1}{2}$. Then $d(\mathcal{P'}) < \frac{n^{2}-2}{2} - \frac{2k}{n^2} = d$ which is a contradiction. Then $|\mathcal{P'}| < 3$.

Suppose that $|\mathcal{P'}| = 1$. Since $G'$ is $(n^2-4)$-regular, its density is $d(\mathcal{P'}) = \frac{n^2-1}{2} - \frac{3}{2} = \frac{n^2-2}{2} - 1 < \frac{n^2-2}{2} - \frac{2k}{n^2} = d$ which is a contradiction. Then $|\mathcal{P'}| > 1$. We conclude that $|\mathcal{P}| = 2$.

Let $A'$ and $B'$ be the two parts of $\mathcal{P}$. Let ${M}_{A'}$, resp. ${M}_{B'}$, be the set of missing edges in $G'[A']$, resp. $G'[B']$. Observe that if $|M_{A'}| + |M_{B'}| \leq k$ then $|M_{A}| + |M_{B}| \leq k$ and then there is a cut of size at least $k$ between $A$ and $B$ in $G$. What it remains to prove is that $|M_{A'}| + |M_{B'}| \leq k$.

As a first step we will show the following inequality we need later $\frac{|{M}_{A'}| + |{M}_{B'}|}{\frac{n^2}{2} + \frac{|{M}_{A'}| + |{M}_{B'}|}{3}} \leq \frac{|{M}_{A'}|}{|A'|} + \frac{|{M}_{B'}|}{|B'|}$. In order to prove this, we consider the following difference 
$$
    \frac{|{M}_{A'}|}{|A'|} + \frac{|{M}_{B'}|}{|B'|} - \frac{|{M}_{A'}|+|{M}_{B'}|}{\frac{|A'|+|B'|}{2} + \frac{|{M}_{A'}| + |{M}_{B'}|}{3}} 
$$

By removing the denominator we get 
$$
    |{M}_{A'}||B'|\left(\frac{|A'| + |B'|}{2} + \frac{|{M}_{A'}|+|{M}_{B'}|}{3}\right) + |{M}_{B'}||A'|\left(\frac{|A'| + |B'|}{2} + \frac{|{M}_{A'}|+|{M}_{B'}|}{3}\right) 
$$    
$$    
    - (|{M}_{A'}| + |{M}_{B'}|)|A'||B'| =
$$     
$$     
     = |{M}_{A'}||B'| \left(\frac{|B'|}{2} + \frac{|{M}_{A'}|}{3} + \frac{|{M}_{B'}|}{3} - \frac{|A'|}{2}\right) + |{M}_{B'}||A'|\left(\frac{|A'|}{2} + \frac{|{M}_{B'}|}{3} + \frac{|{M}_{A'}|}{3} - \frac{|B'|}{2}\right) 
$$

From Lemma~\ref{lem:interneexterne}  for the cubic graph $G_1$ and partition $\{A', B'\}$,  we have $|A'| + \frac{2}{3}|{M}_{B'}| = |B'| + \frac{2}{3}|{M}_{A'}|$, which implies that $\frac{|A'|}{2} = \frac{|B'|}{2} + \frac{|{M}_{A'}|}{3} - \frac{|{M}_{B'}|}{3}$ and $\frac{|B'|}{2} = \frac{|A'|}{2} + \frac{|{M}_{B'}|}{3} - \frac{|{M}_{A'}|}{3}$ and then we get  that the previous equality becomes
$$
    = |{M}_{A'}||B'| \left(\frac{|B'|}{2} + \frac{|{M}_{A'}|}{3} + \frac{|{M}_{B'}|}{3} - \left(\frac{|B'|}{2} + \frac{|{M}_{A'}|}{3} - \frac{|{M}_{B'}|}{3}\right)\right)
$$    
$$    
     + |{M}_{B'}||A'|\left(\frac{|A'|}{2} + \frac{|{M}_{B'}|}{3} + \frac{|{M}_{A'}|}{3} - \left(\frac{|A'|}{2} + \frac{|{M}_{B'}|}{3} - \frac{|{M}_{A'}|}{3}\right)\right) =
$$
$$
    = |{M}_{A'}||B'|\frac{2|{M}_{B'}|}{3} + |{M}_{B'}||A'|\frac{2|{M}_{A'}|}{3}
$$
Since $|{M}_{A'}|$, $|{M}_{B'}|$, $|A'|$ and $|B'|$ are positive integers then  $\frac{|{M}_{A'}|}{|A'|} + \frac{|{M}_{B'}|}{|B'|} - \frac{|{M}_{A'}| + |{M}_{B'}|}{\frac{n^2}{2} + \frac{|{M}_{A'}| + |{M}_{B'}|}{3}} \geq 0$. We conclude that $\frac{|{M}_{A'}| + |{M}_{B'}|}{\frac{n^2}{2} + \frac{|{M}_{A'}| + |{M}_{B'}|}{3}} \leq \frac{|{M}_{A'}|}{|A'|} + \frac{|{M}_{B'}|}{|B'|}$.

\medskip

Finally, we  show that $|{M}_{A'}| + |{M}_{B'}| \leq k$ using the previous inequality. Let $x= |{M}_{A'}| + |{M}_{B'}|$. In order to finalize  the proof, we suppose that $x> k$ and we will arrive at a contradiction, that is   $d(\mathcal{P'}) < d$. Consider the following difference
$$
d - d(\mathcal{P'}) = \frac{n^2 - 2}{2} - \frac{2k}{n^2} - \left(\frac{n^2 - 2}{2} - \frac{|{M}_{A'}|}{|A'|} - \frac{|{M}_{B'}|}{|B'|}\right) = \frac{|{M}_{A'}|}{|A'|} + \frac{|{M}_{B'}|}{|B'|} - \frac{2k}{n^2}
$$
Since $\frac{x}{\frac{n^2}{2} + \frac{x}{3}} \leq \frac{|{M}_{A'}|}{|A'|} + \frac{|{M}_{B'}|}{|B'|}$ then
$$
d - d(\mathcal{P'}) \geq \frac{x}{\frac{n^2}{2}+\frac{x}{3}} - \frac{2k}{n^2} = \frac{x\cdot n^2 - k \cdot n^2 -\frac{2x\cdot k}{3}}{(\frac{n^2}{2}+\frac{x}{3}) \cdot n^2}
$$
Since $x$ and $k$ are integers, then $x\geq k+1$, and by removing the denominator, we get 
$$
    \geq (k+1) \cdot (n^2 - \frac{2}{3}\cdot k) - k\cdot n^2 = n^2 - \frac{2}{3}\cdot k^2 - \frac{2}{3}\cdot k
$$ 

Since $k \leq \frac{n}{2}$ it follows that $n^2 - \frac{2}{3}\cdot k^2 - \frac{2}{3}\cdot k > 0$. This finally gives $d(\mathcal{P'}) < d$, a contradiction to the choice of $\mathcal P'$ as partition with density at least $d$, and we hence conclude that $|{M}_{A'}| + |{M}_{B'}| \leq k$.

Overall, it follows that if $d(\mathcal{P'}) \geq \frac{n^{'}-2}{2} - \frac{2k}{n^2}$ then there is a partition $\{A,B\}$ with an uncut of size at most $k$.
\end{proof}

At the end of this section we show that a partition into a bounded number of cliques provides a good approximation for graphs of large minimum degree.

\begin{theorem}\label{th:delta}
\textsc{Dense Graph Partition} is polynomial-time  $\frac{n-1}{\delta(G)+1}$-approximable on graphs $G$ with $n$ vertices.
\end{theorem}
\begin{proof}
Let $G$  be a graph on $n$ vertices with minimum degree $\delta=\delta(G)$, instance of \MDGP. If $\delta\geq n-3$, we can give an optimum solution in polynomial time by Theorem~\ref{polydense}. So assume $\delta\leq n-4$. 
By Lemma~\ref{lem:densMax}, any partition $\mathcal{P}$ for the vertices of $G$ satisfies $d(\mathcal{P}) \leq \frac{n-1}{2}$. 
Using Brooks' theorem \cite{brooks1941colouring}, $\overline{G}$ is $(n-\delta -1)$-colorable, and further, such a coloring can be computed in polynomial time. (Note that $\delta\leq n-4$ implies that $\overline{G}$ is not a complete graph or a circle, the two exceptions in Brooks' theorem where one more color is needed.) Using such a coloring, $G$ can be partitioned into $n-\delta -1$ cliques. Then the density of this partition is $\frac{n - (n-\delta -1)}{2} = \frac{\delta+1}{2}$. Comparing this value with the upper bound of $\frac{n-1}{2}$ on the optimum shows that this partition into $n-\delta -1$ cliques gives a polynomial-time  $\frac{n-1}{\delta+1}$-approximation for \DGP.
\end{proof}

Notice that if $\delta(G) > \frac{n-3}{2}$, the ratio given in Theorem~\ref{th:delta} improves upon the current best ratio of~2 for \DGP~on general graphs. This approximation can further be used to show the following.

\begin{theorem}
There is an efficient  polynomial-time approximation scheme for \MDGP{} on graphs $G$ with $n$ vertices and $\delta(G)=n-t$, where $t$ is a constant,  $t\geq 4$.
\end{theorem}

\begin{proof}
Let $I=G$  be a graph on $n$ vertices and $\delta(G)=n-t$, instance of \MDGP. We establish in the following an eptas. Given $\varepsilon>0$, consider two cases.

If $n\geq t-1+\frac{t-2}{\varepsilon}$, then
let $\mathcal{P}$ be a partition that corresponds to a $(t-1)$-coloring of $\overline{G}$  such that  each part is a clique in  $G$ as in the proof of Theorem~\ref{th:delta}. Then $d(\mathcal{P}) = \frac{n-t+1}{2}\geq \frac{n+1-\frac{n\varepsilon + \varepsilon+2}{1+\varepsilon}}{2} \geq  \frac{n-1}{2(1+\varepsilon)}\geq \frac{opt(I)}{1+\varepsilon}$, where the last inequality   $opt(I) \leq \frac{n-1}{2}$ comes from Lemma~~\ref{lem:densMax}. 

Otherwise, that is 
$n < t-1+\frac{t-2}{\varepsilon}$, enumerate all the partitions of $G$ and consider the best one. Since the number of partitions of $G$ is the Bell number of order $|V|=n$, $B_n$, and  $B_n\leq n^n$, we get an optimal solution in time $(1/\varepsilon)^ {O(1/\varepsilon)}$. \end{proof}

\section{Conclusion}
  In order to have a better understanding of the complexity of  \MDGP{} it would be nice to study it on other graph classes.  It was proved to be polynomial-time solvable on trees, but the complexity on graphs of bounded treewidth  remains open. Moreover no result exists on split graphs. 
  Concerning approximation, no lower bound was established, it would be nice to improve the 2-approximation algorithm or to show that no polynomial-time approximation scheme exist on general instances. 

\bibstyle{plainurl} 
\bibliography{references.bib}

\end{document}